\documentclass[12pt]{article}

\usepackage[utf8]{inputenc}
\usepackage[letterpaper, margin = 1in]{geometry}

\usepackage{graphicx}
\graphicspath{ {Figures/} }
\usepackage{subcaption}

\usepackage{mathtools, amssymb, bm, amsthm}

\usepackage{booktabs}
\usepackage{enumitem}
\usepackage{cite}
\usepackage{balance}
\usepackage{algorithm}
\usepackage{algorithmic}

\usepackage{hyperref}
\hypersetup{colorlinks = true, linkcolor = blue, citecolor = blue}

\newtheorem{theorem}{Theorem}
\theoremstyle{definition}
\newtheorem{assumption}{Assumption}
\theoremstyle{remark}
\newtheorem{remark}{Remark}

\newcommand{\define}{ \stackrel{\Delta}{=} }
\newcommand{\lra}{ \longrightarrow }

\title{\LARGE{{Multiple Estimation Models for Discrete-time\\Adaptive Iterative Learning Control}}}
\author{Ram Padmanabhan\thanks{Department of Electrical and Computer Engineering and Coordinated Science Laboratory, University of Illinois Urbana-Champaign, Urbana, IL, USA. Email: \texttt{ramp3@illinois.edu.}}, Rajini Makam\thanks{Department of Aerospace Engineering, Indian Institute of Science, Bengaluru, India. Email: \texttt{rajinim@ieee.org.}}, and Koshy George\thanks{Department of Electrical, Electronics and Electronics Engineering, Gandhi Institute of Technology and Management (GITAM), Bengaluru, India. Email: \texttt{kgeorge@ieee.org.}}}
\date{}

\begin{document}

\maketitle

\begin{abstract}
This article focuses on making discrete-time Adaptive Iterative Learning Control more effective using multiple estimation models. Existing strategies use the tracking error to adjust the parametric estimates. Our strategy uses the last component of the identification error to tune these estimates of the model parameters. We prove that this strategy results in bounded estimates of the parameters, and bounded and convergent identification and tracking errors. We emphasize that the proof does not use the Key Technical Lemma. Rather, it uses the properties of square-summable sequences. We extend this strategy to include multiple estimation models and show that all the signals are bounded, and the errors converge. It is also shown that this works whether we switch between the models at every instant and every iteration or at the end of every iteration. Simulation results demonstrate the efficacy of the proposed method with a faster convergence using multiple estimation models.
\end{abstract}

\section{Introduction} \label{sec:Introduction}

Many practical, modern engineering systems require
that a reference trajectory be tracked for a specific finite interval,
and this task is then repeated for multiple iterations. Extensive
research has been dedicated to using Iterative Learning Control (ILC)
for such tasks. In the last two decades, ILC has evolved into a highly
popular strategy to achieve requirements of finite-interval, high
precision tracking control, yet simultaneously maintaining acceptable
levels of control energy \cite{Xu2011, ACM2007, KLM93, BTA2006, BX1998}. These
requirements cannot be achieved satisfactorily by standard
feedback control techniques. In particular, feedback control laws do
not update over iterations, and the error profile is identical
in every iteration. Further, feedback control guarantees only
asymptotic convergence of error, which is unsuitable when considering
finite-interval tracking. High-precision tracking using feedback
control requires prohibitively significant control energy. Each of these
issues is addressed by ILC.

The primary notion in ILC is that performance of a system can be
improved by learning from error and control signals of previous
iterations. Such a notion led to the design of numerous laws that
constructed the control input in iteration $(k+1)$ based directly on
the control input and error signal in iteration $k$. This eventually
developed into a contraction-mapping (CM), operator-theoretic
framework for ILC \cite{KLM93}. Many popular ILC strategies were
designed based on this framework \cite{BTA2006, BX1998, AKM84, CM2002,
  CWS97, CL1996}, and it continues to be popular, with numerous
applications in large-scale industrial manufacturing \cite{WP18,
  AA21}, chemical batch processes \cite{LO2020, LMKL21}, hybrid actuation systems \cite{CWOH2011} and robotics
\cite{OZG2006, CWC2020, CC2020, MA2021}. This framework has also been
extensively analyzed, with established convergence and robustness
results \cite{NG2002, AMC2007, DW1998, LB1996, SBT2022}.

An alternative framework for ILC is based on Composite Energy
Functions (CEFs), which are Lyapunov-like energy functions over
iterations. This framework is particularly useful when system
parameters are unknown, and ILC design must incorporate parameter
estimation. The approach closely follows adaptive control strategies,
with minor differences in parameter update laws \cite{XT2001, XT2002}. CEF-based ILC has
numerous advantages over the contraction-mapping approach to
ILC. First, the restrictive requirement of globally Lipschitz
nonlinearities can be relaxed. French and Rogers \cite{FR2000}
proposed one of the earliest CEF-based techniques to achieve this in
continuous-time Adaptive ILC. Further, it provides a unified method to
address nonlinear systems, systems subjected to disturbances, and
systems with time-varying parameters. This framework can also handle
iteration-varying reference trajectories and random initial conditions
on system states. The CEF approach to continuous-time Adaptive ILC was
formalized in a series of papers \cite{XT2001, XT2002, Xu2002,
  XTL2003, TC2007} in the early 2000s. In contrast to adaptive control, 
a prominent feature of the proposed strategies was the 
\emph{discrete} update of parameter estimates over iterations. Further,
monotonicity of energy functions was demonstrated, resulting in
pointwise convergence of tracking error. Continuous-time Adaptive ILC
has also been extensively applied to robot manipulators \cite{BX1998,
  AT2003, AT2004, LSWT19, WYC19}, high-speed trains \cite{HCMS21,
  YH21, LH21} and vibration control \cite{HM19, FL21}.

Using the analogy between the iteration axis and discrete-time axis, Chi \emph{et al.}
\cite{CHX2008} proposed a discrete-time Adaptive ILC
strategy using Composite Energy Functions for a nonlinear system
subjected to disturbances. The primary features of this strategy
included the ability to deal with iteration-varying reference
trajectories, random initial conditions on the system state and
time-varying system parameters. Applying the Key Technical Lemma (KTL)
\cite{Goodwin} over iterations, the convergence of tracking error was
demonstrated. This technique was further formalized in \cite{CSH2008,
  SLH2012, LZ2015}. In \cite{YZQ2012, YWQ2013, YS2012}, the problem that arises when the sign of the input coefficient is unknown was addressed. In \cite{YHH2016, YL2017} time- and iteration-varying
parameters were both in the problem setup, and a novel
dead-zone approach was proposed to tackle the additional
complexity. Learning control for a system with binary-valued
observations was achieved in \cite{BH18}. A dynamic linearization
framework was used in \cite{YHPD21} for Adaptive ILC on MIMO
systems. In \cite{YH20}, a predictive ILC scheme was used for learning
control of nonaffine, nonlinear systems.

The design and analysis of discrete-time Adaptive ILC closely follow
discrete-time adaptive control. 
Poor transient response is a significant issue in adaptive control and arises from using a single model for parameter estimation. (In what follows, a model is a mathematical representation of the given dynamical system.) A poor initial estimate 
can contribute to large initial tracking and identification
errors. (In this paper, the tracking error is the deviation of the system's state from the reference model state and the identification error is the deviation of the system's state from the state of the estimation model. The latter is defined explicitly in Section \ref{sec:SM-AILC}.) The Multiple Models, Switching and Tuning (MMST) methodology
\cite{NB1992, MM-CT, MM-DT} was proposed to combat this problem. By
initializing many estimation models in the parameter
space, one of these models is likely sufficiently close to
the actual parameter, resulting in improved identification and
tracking performance. Most Adaptive ILC schemes use the tracking error
to update parameter estimates, similar to certain discrete-time
adaptive control strategies. However, such an estimation law does not
lend itself well to the extension to multiple estimation
models. Adaptive control strategies which use the identification error
in place of tracking error for updating parameters have been explored
\cite{MG2021, MG2024}, and these strategies result in improved
convergence with multiple models. The objective of this article is to
present the MMST methodology in the context of discrete-time Adaptive
ILC, by modifying the control and identification laws.

There is very little existing research on using multiple models for
Adaptive ILC, particularly in the CEF framework. In \cite{LW2012, LWL2014}, Li \emph{et al.} present a strategy with multiple fuzzy
neural networks estimating part of the system's parameters. In \cite{FF2015}, Freeman and French use multiple estimation models in
the contraction-mapping setting to present robust stability
and performance-bounds results. In \cite{PBHMG21a, PBHMG21b}, the authors present MMST for Adaptive ILC in a contraction-mapping framework 
and Multiple Models with Second-Level Adaptation (MM-SLA) (see \cite{NH2011, MRG2018}) for Adaptive ILC to achieve 
lower computational complexity. In all the
above cases, however, iteration-varying references cannot be tracked,
and the standard estimation and certainty-equivalent control procedure
in Adaptive ILC are not used. Further, systems subjected to
disturbances are not addressed, and convergence is not demonstrated 
based on composite energy functions. 

In this article, we present a
complete framework for using multiple estimation models in Adaptive
ILC, addressing the above drawbacks. The primary contributions of this
article are as follows:
\begin{itemize}
\item{A new control with a single model
identification scheme are presented for discrete-time Adaptive ILC, and
the identification error (rather than tracking error) is used to
update parameter estimates.} 
\item{Convergence is proved using CEFs and the
properties of square-summable sequences rather than using the KTL.} 
\item{Next, a strategy with multiple estimation
models based on MMST is proposed. A complete overview of the control
and identification laws is provided, two switching schemes are
outlined, and convergence is proved in a unified manner for both
schemes, using the properties of square-summable sequences.} 
\item{Simulation results indicate that both the single model and multiple model
estimation schemes demonstrate satisfactory tracking performance. The multiple 
model scheme results in faster convergence of tracking
errors for linear time-invariant and time-varying systems, as well as nonlinear, discrete-time systems subjected to disturbances.}
\end{itemize}

The remainder of this article is organized as follows. In Section
\ref{sec:Formulation}, the general discrete-time Adaptive ILC problem
is introduced, with standard assumptions and remarks. Section
\ref{sec:SM-AILC} presents a new single estimation model solution to
this problem. This is extended to a strategy with multiple estimation
models in Section \ref{sec:MM-AILC}. We present simulation results for
the proposed strategies in Section \ref{sec:Simulations}, and
concluding remarks in Section \ref{sec:Conclusion}.

Throughout this article, $\mathbb{N}$ denotes the set of natural 
numbers $\{1,
2, \ldots\}$, and $\mathbb{R}^n$ denotes the vector space of all $n$-tuples of
real numbers. For a vector $y = \left[y_1, \ldots, y_n\right]^T \in
\mathbb{R}^n$, $\|y\|$ denotes the Euclidean norm, defined as
$$
\|y\| \define \sqrt{\sum_{i = 1}^{n} y_{i}^{2}}. 
$$
$\ell_2$ denotes the Hilbert space of all square-summable sequences,
i.e. all sequences $\{x_n\}_{n \in \mathbb{N}}$ such that
$$
\sum_{n} |x_{n}|^2 < \infty, 
$$
and $\ell_{\infty}$ denotes the Banach space of all bounded sequences,
i.e. all sequences $\{x_n\}_{n \in \mathbb{N}}$ such that
$$
\sup_n |x_n| < \infty.
$$

\section{Problem Formulation} \label{sec:Formulation}

In this section, we formulate the general discrete-time Adaptive ILC
problem. Consider the following discrete-time, nonlinear, uncertain
$n$th order system with matched, time-varying uncertainty:
\begin{align}
x_{i, k}(t+1) &= x_{i+1, k}(t), \hspace{0.5cm} i = 1, \ldots, (n-1),
\nonumber \\ 
x_{n, k}(t+1) &= \theta_{1}^{T}(t)\xi(x_k(t)) + b(t)u_k(t) +
d(t), \label{eq:Plant1} 
\end{align}
where $k \in \mathbb{N}$ denotes the index for iterations and each
iteration consists of samples indexed $\{0, 1, \ldots, T\}$. The time
index $t$ is within the set $\mathcal{I}_T = \{0, 1, \ldots,
T-1\}$. Note that this set does not include the final sample
$T$. $x_{i, k}(t)$ denotes state $i$ in iteration $k$ at sample
$t$. $x_k(t) = \left[x_{1, k}(t), \ldots, x_{n, k}(t)\right]^T \in
\mathbb{R}^n$ is the measurable state vector. $\theta_1(t) \in
\mathbb{R}^p$ is an unknown parameter vector, $b(t) \in \mathbb{R}$ is
the unknown input coefficient, and $d(t) \in \mathbb{R}$ is an unknown
exogenous disturbance. Each of these quantities is iteration
invariant. $\xi(x_k(t)) \equiv \xi_k(t) \in \mathbb{R}^p$, called the
\emph{regression vector}, is a known, bounded
nonlinear vector function of the
state $x_k(t)$, and $u_k(t) \in \mathbb{R}$ is the input to the system
in iteration $k$ and at sample $t$. Equation \eqref{eq:Plant1} can be
rewritten as follows:
\begin{align}
x_{i, k}(t+1) &= x_{i+1, k}(t), \hspace{0.5cm} i = 1, \ldots, (n-1), \nonumber \\
x_{n, k}(t+1) &= \theta^{T}(t)\phi_k(t), \label{eq:Plant}
\end{align}
where $\theta(t) \define \left[\theta_{1}^{T}(t), b(t), d(t)\right]^T
\in \mathbb{R}^{p+2}$ is the overall unknown parameter vector, and
$\phi_k(t) \define \left[\xi_{k}^{T}(t), u_k(t), 1\right]^T \in
\mathbb{R}^{p+2}$ is the overall known regression vector.

The objective of the Adaptive ILC problem is to design an appropriate
control input $u_k(t)$ such that the system state $x_k(t)$ tracks the
state $x_{m, k}(t)$ of the following stable, iteration-varying
reference model:
\begin{align}
x_{i, m, k}(t+1) &= x_{i+1, m, k}(t) \hspace{0.5cm} i = 1, \ldots,
(n-1), \nonumber \\ 
x_{n, m, k}(t+1) &= \rho_k(t), \label{eq:RefModel}
\end{align}
for some known $\rho_k(t)$, with asymptotic tracking over iterations
$k$. Defining the state tracking error $e_k(t) \define x_k(t) - x_{m,
  k}(t) = \left[e_{1, k}(t), \ldots, e_{n, k}(t)\right]^T \in
\mathbb{R}^n$, asymptotic tracking over iterations implies:
\begin{equation} \label{eq:Objective}
\lim_{k \lra \infty} e_k(t) = 0
\end{equation}
for each sample $t \in \mathcal{I}_T$.

The following assumptions are made:
\begin{assumption} \label{asm:Bounded}
The unknown quantities $\theta_1(t)$, $b(t)$ and $d(t)$ are bounded,
and hence the parameter vector $\theta(t)$ is bounded.
\end{assumption}
\begin{assumption} \label{asm:Sign}
The sign of $b(t)$ is known and invariant, i.e. $b(t)$ is either
positive or negative for all time $t$, and $b(t)$ is
non-singular. Without loss of generality, assume $b(t) \geq
b_{\mathrm{min}} > 0$. This assumption implies that the control
direction is known.
\end{assumption}

\begin{remark}
Throughout this article, the case with iteration invariant parameters
$\theta_1(t)$, $b(t)$ and $d(t)$, and hence iteration invariant
$\theta(t)$ is considered. This can be extended to the case with time-
and iteration-varying parameters, which was addressed in \cite{YL2017} using a novel dead-zone approach. Applying this to the
proposed techniques is an interesting avenue for future work on
this topic.
\end{remark}

\begin{remark}
Throughout this article, we do not assume identical initial conditions
on the plant and reference model. However, it has been observed in
\cite{CHX2008} that with random, non-zero initial conditions on the
plant \eqref{eq:Plant}, random non-zero initial errors will propagate,
and the state errors $e_{i, k}(t)$ for $i = 1, \ldots, (n-1)$ and $t =
0, \ldots, (n-i)$ cannot be `learned', as these errors are not
affected by the input $u_k(t)$. The remaining errors are dependent on
$u_k(t)$ and hence can be driven to zero. If the plant and reference
model have identical initial conditions, it can be shown that each
component of the identification and tracking error vectors converge to
zero.
\end{remark}

\begin{remark}
Assumption \ref{asm:Sign} can be relaxed by employing the technique of
discrete Nussbaum gain \cite{LN1986}, as explored in \cite{YZQ2012, YWQ2013, QGX2022}. An alternative approach without using the Nussbaum gain was
also explored in \cite{YS2012} by fully exploiting the convergence
properties of parameter estimates and incorporating two modifications
in the control and parameter update laws. Extending the techniques
proposed here by incorporating these approaches when the control
direction is unknown is a promising avenue for future work.
\end{remark}

\begin{remark}
In numerous works on discrete-time Adaptive ILC, an additional
assumption --- usually called the linear growth rate or sector-bounded
condition --- is made. This assumption states that the nonlinearity
$\xi_k(t)$ satisfies:
$$ \left\|\xi_k(t)\right\| \leq c_1 + c_2\left\|x_k(t)\right\|, $$
for some positive constants $c_1$ and $c_2$. This assumption plays a
key role in analyzing the convergence of the tracking error over
iterations as part of the assumptions for the KTL
\cite{Goodwin}. In contrast, our analysis does not involve the KTL, and hence we do not make this assumption.
\end{remark}

\section{A New Solution for Discrete-time Adaptive ILC} \label{sec:SM-AILC}

In this section, we formulate a new control law and parameter update
law for the problem formulated in Section
\ref{sec:Formulation}. Existing solutions mainly incorporate the
principle of certainty equivalence for control design and use the
tracking error for estimating and updating parameters. The
disadvantage of using the tracking error is that this strategy cannot
be extended to the use of multiple estimation models. Using the
identification error in place of tracking error has been explored for
discrete-time adaptive control in \cite{MG2021, MG2024}, with
results demonstrating improved convergence with multiple estimation
models. Further, the stability proofs do not invoke the KTL and instead use Lyapunov theory and the properties of
square-summable, or $\ell_2$ sequences. Here, we use the analogy
between the discrete-time and iteration axes to formulate a
corresponding Adaptive ILC strategy.

\subsection{Control and Identification Laws} \label{sec:SM-1}

Construct an identification model with state $\hat{x}_k(t)$:
\begin{align}
\hat{x}_{i, k}(t+1) &= \hat{x}_{i+1, k}(t), \hspace{0.5cm} i = 1,
\ldots (n-1), \nonumber \\ 
\hat{x}_{n, k}(t+1) &= \hat{\theta}_{1, k}^{T}(t)\xi(x_k(t)) +
\hat{b}_k(t)u_k(t) + \hat{d}_k(t) \nonumber \\ 
&= \hat{\theta}_{k}^{T}(t)\phi_k(t). \label{eq:IdentModelSM}
\end{align}
The purpose of establishing this identification model is to construct an identification error that is used to estimate the unknown parameter vector $\theta(t)$. $\hat{\theta}_{1, k}(t)$, $\hat{b}_k(t)$ and $\hat{d}_k(t)$ denote the
estimates of the quantities $\theta_1(t)$, $b(t)$ and $d(t)$ in
iteration $k$ and sample $t$. Correspondingly, $\hat{\theta}_{k}(t)$
denotes the estimate of the parameter vector $\theta(t)$. Define
$\Tilde{\theta}_k(t) \define \theta(t) - \hat{\theta}_k(t)$. Further, define the state identification error $\hat{e}_{k}(t) \define x_k(t) -
\hat{x}_k(t)$. Then, from \eqref{eq:Plant} and
\eqref{eq:IdentModelSM},
\begin{align}
\hat{e}_{i, k}(t+1) &= \hat{e}_{i+1, k}(t), \hspace{0.5cm} i = 1,
\ldots, (n-1), \nonumber \\ 
\hat{e}_{n, k}(t+1) &=
\Tilde{\theta}_{k}^{T}(t)\phi_k(t). \label{eq:IdentErrSM} 
\end{align}
Finally, the tracking error $e_k(t) = x_k(t) - x_{m, k}(t)$ can be
described by: 
\begin{align}
e_{i, k}(t+1) &= e_{i+1, k}(t), \hspace{0.5cm} i = 1, \ldots, (n-1),
\nonumber \\ 
e_{n, k}(t+1) &= \theta_{1}^{T}(t)\xi_k(t) + b(t)u_k(t) + d(t) -
\rho_k(t). \label{eq:TrackErr} 
\end{align}
Using \eqref{eq:TrackErr}, the following control
law is generated:
\begin{equation} \label{eq:ControlSM}
u_k(t) = \frac{1}{\hat{b}_{k}(t)}\left[\beta e_{n, k-1}(t+1) +
  \rho_k(t) - \hat{\theta}_{1, k}^{T}(t)\xi_k(t) - \hat{d}_k(t)\right],
\end{equation}
where $0 < \beta < 1$. Add and subtract $\hat{b}_k(t)u_k(t)$ in \eqref{eq:TrackErr}, substitute \eqref{eq:ControlSM} in
\eqref{eq:TrackErr} and use \eqref{eq:IdentErrSM}:
\begin{align}
e_{n, k}(t+1) &= \Tilde{\theta}_{k}^{T}(t)\phi_k(t) + \beta e_{n,
  k-1}(t+1) \nonumber \\ 
&= \hat{e}_{n, k}(t+1) + \beta e_{n, k-1}(t+1). \label{eq:SMError}
\end{align}
The presence of $\beta$ in the control law is to provide some damping
by incorporating previous iteration errors. Existing Adaptive ILC
schemes set $\beta = 0$, with no previous iteration tracking error
term, resulting in a deadbeat-like law.

The parameter vector estimate is updated according to the following
adaptive law:
\begin{equation} \label{eq:ParEstSM}
\hat{\theta}_{k+1}(t) = \mathrm{Proj}\left[\hat{\theta}_k(t) +
  \frac{\phi_k(t)}{1 + \left\|\phi_k(t)\right\|^2}\hat{e}_{n,
    k}(t+1)\right]. 
\end{equation}
Note that this law uses the identification error, in contrast to
existing Adaptive ILC strategies that use the tracking error for
updating parameters. This law is similar to the projection algorithm
widely used in adaptive control \cite{Goodwin}, except with the update
over iterations rather than time. The error $\hat{e}_{n, k}(t+1)$ is available
as the update is performed offline at the end of iteration $k$. The
operator $\mathrm{Proj}[.]$ is defined below. Define a vector $m$ as
follows:
$$
m \define \left[\hat{\theta}_k(t) + \frac{\phi_k(t)}{1 +
    \left\|\phi_k(t)\right\|^2}\hat{e}_{n, k}(t+1)\right] =
\left[m_{1}^{T}, m_2, m_3\right]^T, 
$$
where $m_{1} \in \mathbb{R}^p$, $m_2 \in \mathbb{R}$ and $m_3 \in \mathbb{R}$ denote the estimates of $\theta_1(t)$,
$b(t)$ and $d(t)$ prior to projection. Then,
\begin{equation} \label{eq:Projection}
\mathrm{Proj}[m] \define
\begin{dcases}
\left[m_{1}^{T}, m_2, m_3\right]^T & \text{if } m_2 \geq b_{\mathrm{min}} \\
\left[m_{1}^{T}, b_{\mathrm{min}}, m_3\right]^T & \text{if } m_2 < b_{\mathrm{min}}
\end{dcases}
.
\end{equation}
The use of the projection operator defined above ensures that division
by zero is avoided in the control law \eqref{eq:ControlSM}.

\begin{remark}
Throughout this article, the time index $t$ is always in the set
$\mathcal{I}_T$. The control \eqref{eq:ControlSM} and adaptive
\eqref{eq:ParEstSM} laws are defined on this time horizon. However,
note that all state variables and errors are formulated on the time
horizon $\{1, \ldots, T\}$, apart from their initial
conditions. Hence, state variables and errors use the index $(t+1)$
throughout, as evident from the control and adaptive laws
above.  Further, note that neither the control law nor the adaptive law
is defined at the final sample $T$. However, they affect the state variables 
and errors corresponding to this sample.
\end{remark}

\subsection{Convergence Analysis} \label{sec:SM-2}

We have the following result for convergence of the proposed Adaptive
ILC law:
\begin{theorem} \label{thm:SM}
For the system \eqref{eq:Plant} with the objective of tracking the
reference model \eqref{eq:RefModel}, the control law
\eqref{eq:ControlSM} along with the adaptive law \eqref{eq:ParEstSM}
guarantees the following:
\begin{enumerate}[label = { $\arabic*$. }]
\item{$\Tilde{\theta}(t), \hat{\theta}(t) \in \ell_\infty$ for each $t
  \in \mathcal{I}_T$, i.e. the sequence of parametric errors
  $\Tilde{\theta}_k(t)$ over iterations --- and hence the sequence of
  parameter estimates $\hat{\theta}_k(t)$ over iterations --- is
  bounded for each sample $t$.} 
\item{$\hat{e}_{n}(t+1) \in \ell_2\cap\ell_\infty$ for each $t \in
  \mathcal{I}_T$, i.e. the sequence of the $n$\text{\normalfont{th}}
  component of the identification error over iterations is
  square-summable and bounded for each sample $t$.} 
\item{With identical initial conditions on the plant and reference
  model, $\lim_{k \lra \infty} \hat{e}_k(t+1) = 0$ for each $t \in
  \mathcal{I}_T$, i.e. each component of the identification error
  vector tends to zero with iterations, for each sample $t$.} 
\item{With identical initial conditions on the plant and reference
  model, $\lim_{k \lra \infty} e_k(t+1) = 0$ for each $t \in
  \mathcal{I}_T$, i.e. each component of the tracking error vector
  tends to zero with iterations, for each sample $t$.} 
\item{$\lim_{k \lra \infty} \left\|\hat{\theta}_k(t) -
  \hat{\theta}_{k-p}(t)\right\|^2 = 0$, for each $t \in
  \mathcal{I}_T$, for any $p \in \mathbb{N}$, i.e. the parameter
  vector estimates converge over iterations for each sample $t$.} 
\end{enumerate}
\end{theorem}

\begin{proof}
The proof is organized into three parts. Part $1$ derives the
boundedness of $\Tilde{\theta}(t)$, Part $2$ demonstrates that all
errors converge to zero over iterations, and Part $3$ shows that
parameter vector estimates converge over iterations. Thus, statement
$1$ of the theorem is proved in Part $1$, statements $2$, $3$ and $4$
are proved in Part $2$, and statement $5$ is proved in Part $3$.

\vspace{0.5em} \noindent \emph{Part 1: Boundedness of Parametric Error:}

\noindent
Define a composite energy function (CEF) $V_k(t)$:
\begin{equation} \label{eq:CEFSM}
V_k(t) \define \Tilde{\theta}_{k}^{T}(t)\Tilde{\theta}_{k}(t) =
\left\|\Tilde{\theta}_k(t)\right\|^2. 
\end{equation}
Let $\Delta V_k(t) \define V_{k+1}(t) - V_k(t)$. Then, from
\eqref{eq:ParEstSM}, 
\begin{equation} \label{eq:L1}
\Delta V_k(t) = \Big\|\theta(t) - \mathrm{Proj}[m]\Big\|^2 -
\left\|\Tilde{\theta}_k(t)\right\|^2. 
\end{equation}
Consider the scalar $\left|b(t) - \mathrm{Proj}[m_2]\right|$, and note
that $b(t) \geq b_{\mathrm{min}}$. 
\begin{itemize}
\item{When $m_2 \geq b_{\mathrm{min}}$, $\mathrm{Proj}[m_2] =
  m_2$. Then, $\left|b(t) - \mathrm{Proj}[m_2]\right| =  \left|b(t) -
  m_2\right|$.}  
\item{When $m_2 < b_{\mathrm{min}} \leq b(t)$, $\mathrm{Proj}[m_2] =
  b_{\mathrm{min}}$. Then, $\left|b(t) - \mathrm{Proj}[m_2]\right| =
  \left|b(t) - b_{\mathrm{min}}\right| < \left|b(t) - m_2\right|$.} 
\end{itemize}
Thus, the relation $\left|b(t) - \mathrm{Proj}[m_2]\right| \leq
\left|b(t) - m_2\right|$ always holds. As $b(t)$ is simply part of the
parameter vector $\theta(t)$, the relation $\left\|\theta(t) -
\mathrm{Proj}[m]\right\| \leq \left\|\theta(t) - m\right\|$ always
holds. Hence, the parametric error magnitude does not increase
using the projection operator. Using this in \eqref{eq:L1},
\begin{equation} \label{eq:L2}
\Delta V_k(t) \leq \Big\|\theta(t) - m\Big\|^2 -
\left\|\Tilde{\theta}_k(t)\right\|^2. 
\end{equation}
Substitute $\displaystyle m = \left[\hat{\theta}_k(t) +
  \frac{\phi_k(t)}{1 + \left\|\phi_k(t)\right\|^2}\hat{e}_{n,
    k}(t+1)\right]$. Then,
\begin{equation} \label{eq:L3}
\Delta V_k(t) \leq \left\|\Tilde{\theta}_k(t) - \frac{\phi_k(t)}{1 +
  \left\|\phi_k(t)\right\|^2}\hat{e}_{n, k}(t+1)\right\|^2 -
\left\|\Tilde{\theta}_k(t)\right\|^2.
\end{equation}
On simplification by expanding the norm and using \eqref{eq:IdentErrSM}, this reduces to:
\begin{equation} \label{eq:DeltaVSM}
\Delta V_k(t) \leq - \left(\frac{2 +
  \left\|\phi_k(t)\right\|^2}{\left(1 +
  \left\|\phi_k(t)\right\|^2\right)^2}\right)\hat{e}_{n, k}^{2}(t+1), 
\end{equation}
or,
\begin{equation} \label{eq:DecreasingSM}
\Delta V_k(t) \leq - \alpha_{k}^{2}(t) \hat{e}_{n, k}^{2}(t+1) \leq 0,
\end{equation}
where $\alpha_{k}^{2}(t)$ denotes the positive quantity in parentheses
in \eqref{eq:DeltaVSM}. Thus, the function $V_k(t)$ is
non-increasing. From this and the construction of $V_k(t)$
\eqref{eq:CEFSM}, it is evident that $\Tilde{\theta}_k(t)$ is a
bounded sequence over iterations $k$, for every $t \in \mathcal{I}_T$,
i.e. $\Tilde{\theta}(t) \in \ell_\infty$. Subsequently, as $\theta(t)$
is bounded, the sequence of parameter estimates $\hat{\theta}_k(t)$
over iterations $k$ is bounded, i.e. $\hat{\theta}(t) \in
\ell_\infty$. This concludes Part $1$ of the proof.

\vspace{0.5em} \noindent \emph{Part 2: Convergence of Errors:}

\noindent
From \eqref{eq:DecreasingSM}, note that $\lim_{N \lra \infty}
\left|V_{N+1}(t) - V_1(t)\right| < \infty$ for each $t \in
\mathcal{I}_T$. This can be written as:
$$
\lim_{N \lra \infty} \left|\sum_{k = 1}^{N} \Delta V_k(t)\right| \leq
\lim_{N \lra \infty} \sum_{k = 1}^{N} \alpha_{k}^{2}(t)\hat{e}_{n,
  k}^{2}(t+1) < \infty.  $$
Then, the sequence $\alpha_k(t)\hat{e}_{n, k}(t+1)$ is square-summable
over iterations $k$, for each $t \in \mathcal{I}_T$. By the properties
of $\ell_2$ sequences, $\alpha(t)\hat{e}_n(t+1) \in
\ell_2\cap\ell_\infty$. 
Next, note that since $\xi_k(t)$ and $u_k(t)$ are bounded, so is 
$\left\|\phi_k(t)\right\|$ by definition. Hence,
$ \alpha_k(t)$  can never converge to $0$. 
Further,  $0 < \alpha_k(t) < \sqrt{2}$ and we have $\hat{e}_n(t+1) \in
\ell_2\cap\ell_\infty$, and hence,
\begin{equation} \label{eq:ConvIdentErrSM}
\lim_{k \lra \infty} \hat{e}_{n, k}(t+1) = 0
\end{equation}
for each $t \in \mathcal{I}_T$. Further, consider
eq. \eqref{eq:SMError}. This is an iteration-domain difference
equation, with a forcing function $\hat{e}_{n, k}(t+1) \lra 0$ as $k
\lra \infty$. As $0 < \beta < 1$, it is evident that:
\begin{equation} \label{eq:ConvTrackErrSM}
\lim_{k \lra \infty} e_{n, k}(t+1) = 0
\end{equation}
for each $t \in \mathcal{I}_T$. Finally, under the assumption of
identical initial conditions, \eqref{eq:ConvIdentErrSM} and
\eqref{eq:ConvTrackErrSM} imply that:
\begin{equation} \label{eq:ConvIdentErrVecSM}
\lim_{k \lra \infty} \hat{e}_k(t+1) = 0
\end{equation}
and
\begin{equation} \label{eq:ConvTrackErrVecSM}
\lim_{k \lra \infty} e_k(t+1) = 0
\end{equation}
for each $t \in \mathcal{I}_T$, i.e. the identification and tracking
error \emph{vectors} converge to $0$ as $k \lra \infty$. This
concludes Part $2$ of the proof.

\vspace{0.5em} \noindent \emph{Part 3: Convergence of Parameter Estimates:}

\noindent
Consider the update law \eqref{eq:ParEstSM}, and consider the scalar
$\left|\mathrm{Proj}[m_2] - \hat{b}_k(t)\right|$. From the preceding
update of parameter estimates, $\hat{b}_k(t) \geq b_{\mathrm{min}}$,
by the use of projection.
\begin{itemize}
\item{When $m_2 \geq b_{\mathrm{min}}$, $\mathrm{Proj}[m_2] = m_2$,
  and $\left|\mathrm{Proj}[m_2] - \hat{b}_k(t)\right| = \left|m_2 -
  \hat{b}_k(t)\right|$.} 
\item{When $m_2 < b_{\mathrm{min}} \leq \hat{b}_k(t)$,
  $\mathrm{Proj}[m_2] = b_{\mathrm{min}}$. Then,
  $\left|\mathrm{Proj}[m_2] - \hat{b}_k(t)\right| =
  \left|b_{\mathrm{min}} - \hat{b}_k(t)\right| < \left|m_2 -
  \hat{b}_k(t)\right|$.} 
\end{itemize}
Thus, the relation $\left|\mathrm{Proj}[m_2] - \hat{b}_k(t)\right|
\leq \left|m_2 - \hat{b}_k(t)\right|$ always holds. By extension, the
relation 
$$
\left\|\mathrm{Proj}[m] - \hat{\theta}_k(t)\right\| \leq
\left\|m - \hat{\theta}_k(t)\right\|
$$ 
always holds. Using
\eqref{eq:ParEstSM} and substituting $m$,
$$
\left\|\hat{\theta}_{k+1}(t) - \hat{\theta}_k(t)\right\|^2 \leq
\left\|\frac{\phi_k(t)}{1 + \left\|\phi_k(t)\right\|^2}\hat{e}_{n,
  k}(t+1)\right\|^2 \leq \hat{e}_{n, k}^{2}(t+1).
$$
Thus,
$$
\lim_{N \lra \infty} \sum_{k = 1}^{N} \left\|\hat{\theta}_{k+1}(t) -
\hat{\theta}_k(t)\right\|^2 \leq \lim_{N \lra \infty} \sum_{k = 1}^{N}
\hat{e}_{n, k}^2(t+1) < \infty, $$
or,
\begin{equation} \label{eq:ConvParEst1SM}
\lim_{k \lra \infty} \left\|\hat{\theta}_{k+1}(t) -
\hat{\theta}_k(t)\right\| = 0 
\end{equation}
for each $t \in \mathcal{I}_T$. Thus, parameter estimates one
iteration apart converge. This result can easily be extended to the
difference between parameter estimates $p$ iterations apart, as
follows:
\begin{align*}
&\left\|\hat{\theta}_{k}(t) - \hat{\theta}_{k-p}(t)\right\| \\
\hspace{0.3cm} & = \left\|\hat{\theta}_k(t) - \hat{\theta}_{k-1}(t) +
  \hat{\theta}_{k-1}(t) - \hat{\theta}_{k-2}(t) + \ldots + \hat{\theta}_{k-p+1}(t) -
  \hat{\theta}_{k-p}(t)\right\| \\ 
& \leq \left\|\hat{\theta}_k(t) -
  \hat{\theta}_{k-1}(t)\right\| + \left\|\hat{\theta}_{k-1}(t) -
  \hat{\theta}_{k-2}(t)\right\| + \ldots + \left\|\hat{\theta}_{k-p+1}(t) -
  \hat{\theta}_{k-p}(t)\right\|. 
\end{align*}
Thus, taking the limit as $k \lra \infty$ on both sides,
$$
\lim_{k \lra \infty} \left\|\hat{\theta}_{k}(t) -
\hat{\theta}_{k-p}(t)\right\| \leq \lim_{k \lra \infty}
\left\|\hat{\theta}_k(t) - \hat{\theta}_{k-1}(t)\right\| + \ldots + \lim_{k \lra \infty} \left\|\hat{\theta}_{k-p+1}(t) -
\hat{\theta}_{k-p}(t)\right\|. 
$$
As each limit on the right is zero from \eqref{eq:ConvParEst1SM}, and
the norm is always non-negative,
\begin{equation} \label{eq:ConvParEstPSM}
\lim_{k \lra \infty} \left\|\hat{\theta}_{k}(t) - \hat{\theta}_{k-p}(t)\right\| = 0,
\end{equation}
for each $t \in \mathcal{I}_T$, i.e. parameter estimates converge over
iterations. This concludes the proof of Theorem \ref{thm:SM}.
\end{proof}

In summary, this section has presented a new approach to solving the
discrete-time Adaptive ILC problem. A new control law with an
additional scaled tracking error term is formulated, and parameter
estimates are updated using the identification error rather than the
tracking error. It is then proved that each component of the identification and tracking
error vectors converges to $0$ with iterations $k$. The proof of
convergence does not involve the KTL. Instead, simple
inferences from the non-increasing nature of $V_k(t)$ are used
concurrently with properties of $\ell_2$ sequences. The approach
presented in this section also enables the extension to the multiple
estimation models case, as described in the following section.

\section{Multiple Estimation Models for Adaptive ILC} \label{sec:MM-AILC}

Adaptive control strategies can suffer from the poor transient performance
of identification, tracking and parametric errors when a single model
is used for parameter estimation. In particular, the
initial parametric uncertainty is likely large, contributing
significantly to poor transient response. The methodology of Multiple 
Models, Switching and Tuning (MMST) was proposed to address this 
issue \cite{NB1992, MM-CT, MM-DT}. The methodology in discrete-time
adaptive control is as follows. A number of models (say $M$) are
initialized in the parameter space with different initial
conditions. Each of these is updated according to standard parameter
estimation algorithms \cite{Goodwin} every sample. At each sample, one
model is chosen according to a criterion, and the parameter estimates
corresponding to that model are used for control design. The most
common criterion used is a minimum identification error criterion,
stated as follows. At each sample, pick the model $j^*$ that satisfies
$j^* = \mathrm{arg} \min_{j = 1, \ldots, M}
\left|\hat{e}_j(t)\right|$, where $\hat{e}_j(t)$ denotes the
identification error corresponding to model $j$ at time instant $t$.

As mentioned in Section \ref{sec:Introduction}, there is very little
existing research on the use of the MMST methodology in Adaptive
ILC. This section presents the main results of this article, designing
a general approach to using multiple models in Adaptive ILC, proposing
two strategies for switching between models and proving convergence
in both cases. Section \ref{sec:MM-1} describes the formulation of
control and identification laws for the proposed strategies, and
Section \ref{sec:MM-2} presents the proof of convergence of the
identification and tracking errors.

\subsection{Control and Identification Laws} \label{sec:MM-1}

The basic formulation of the problem remains the same as described in
Section \ref{sec:Formulation}. However, instead of a single
identification model \eqref{eq:IdentModelSM} as in Section
\ref{sec:SM-AILC}, we construct $M$ identification models. Let
$\mathcal{M} = \left\{1, \ldots, M\right\}$ denote the set of model
indices. Then, each model has a state $\hat{x}_{j, k}(t)$, $j \in
\mathcal{M}$, that evolves as follows:
\begin{align}
\hat{x}_{i, j, k}(t+1) &= \hat{x}_{i+1, j, k}(t), \hspace{0.5cm} i =
1, \ldots (n-1), \nonumber \\ 
\hat{x}_{n, j, k}(t+1) &= \hat{\theta}_{1, j, k}^{T}(t)\xi(x_k(t)) +
\hat{b}_{j, k}(t)u_k(t) + \hat{d}_{j, k}(t) \nonumber \\ 
&= \hat{\theta}_{j, k}^{T}(t)\phi_k(t). \label{eq:IdentModelMM}
\end{align}
$\hat{x}_{i, j, k}(t)$ denotes state $i$ of identification model $j$
in iteration $k$ and sample $t$, and $\hat{\theta}_{j, k}(t)$ denotes
the parameter estimate of model $j$ in iteration $k$ and sample
$t$. Define the $M$ parametric errors $\Tilde{\theta}_{j, k}(t)
\define \theta(t) - \hat{\theta}_{j, k}(t)$, and the $M$
identification errors $\hat{e}_{j, k}(t) \define x_k(t) - \hat{x}_{j,
  k}(t)$. Using \eqref{eq:Plant} and \eqref{eq:IdentModelMM},
\begin{align}
\hat{e}_{i, j, k}(t+1) &= \hat{e}_{i+1, j, k}(t), \hspace{0.5cm} i =
1, \ldots, (n-1), \nonumber \\ 
\hat{e}_{n, j, k}(t+1) &= \Tilde{\theta}_{j,
  k}^{T}(t)\phi_k(t). \label{eq:IdentErrMM} 
\end{align}
As before, the tracking error is described by eq. \eqref{eq:TrackErr}. 

We are now presented with two options:
\subsubsection{Case 1} Continue using the analogy between the
discrete-time axis in adaptive control and the iteration axis in
Adaptive ILC, and switch between models only once every iteration, at
the end. The criterion for switching is then chosen as: 
\begin{equation} \label{eq:Case1}
	j_{k}^{*} = \mathrm{arg} \min_{j \in \mathcal{M}}
        \left[\sum_{t \in \mathcal{I}_T}\left|\hat{e}_{n, j,
            k-1}(t+1)\right|^2\right],
\end{equation}
i.e. the model producing minimum energy in the $n$th component of the
identification error (and hence minimum energy in the identification
error vector) in iteration $(k-1)$ is chosen for control design in iteration
$k$.

\subsubsection{Case 2} Switch between models at every sample $t$ in
every iteration $k$. The criterion for switching is then chosen as: 
\begin{equation} \label{eq:Case2}
	j_{k}^{*}(t) = \mathrm{arg} \min_{j \in \mathcal{M}}
        \left|\hat{e}_{n, j, k}(t)\right|,
\end{equation}
i.e. at every sample, a new model is chosen based on the minimum
identification error at that sample, and is used for control design
at that sample.

\begin{remark}
In Case $2$, as identification error is on the time horizon $\left\{1,
\ldots, T\right\}$, so is the sequence of models
$j_{k}^{*}(t)$. However, the control design is on the horizon $t \in
\mathcal{I}_T$. Hence, the final model chosen, $j_{k}^{*}(T)$, is used
for designing $u_{k+1}(0)$, the initial control input of the next
iteration.
\end{remark}

Note how the best model $j^*$ depends only on iteration $k$ in Case
$1$, but depends on both iteration $k$ and time $t$ in Case $2$. The
control law can then be formulated as follows:
\begin{equation} \label{eq:ControlMM}
u_k(t) = \frac{1}{\hat{b}_{j^*, k}(t)}\left[\beta e_{n, k-1}(t+1) +
  \rho_k(t) - \hat{\theta}_{1, j^*, k}^{T}(t)\xi_k(t) - \hat{d}_{j^*, k}(t)\right], 
\end{equation}
where $0 < \beta < 1$, and $j^*$ denotes either $j_{k}^{*}$ or
$j_{k}^{*}(t)$, depending on whether criterion \eqref{eq:Case1} or
\eqref{eq:Case2} is being used. Evidently, the control law uses
parameter estimates corresponding to the model with minimum
identification error in the sense of either criterion. In iteration $1$,
model $1$ is chosen for control design without loss of generality. Note that all
models continue to be updated irrespective of which model is chosen in
\eqref{eq:ControlMM}. Substituting \eqref{eq:ControlMM} in
\eqref{eq:TrackErr}, and using \eqref{eq:IdentErrMM},
\begin{align}
e_{n, k}(t+1) &= \Tilde{\theta}_{j^*, k}^{T}(t)\phi_k(t) + \beta e_{n,
  k-1}(t+1) \nonumber \\ 
&= \hat{e}_{n, j^*,  k}(t+1) + \beta e_{n, k-1}(t+1). \label{eq:MMError}
\end{align}

\begin{algorithm}[!t] 
\caption{Computational flow with Multiple Models Case $1$.}
\hrule \vspace*{0.1in}
\begin{algorithmic} 
\STATE { \bf Initialisation}: 
$\hat{\theta}_{j,0}(t) \leftarrow \mathrm{random}$, $j \in \mathcal{M}$, $t \in \mathcal{I}_T$; $j_{1}^{*} = 1$.
\FOR{$k = 1, 2, \ldots$}
\FOR{$t = 0, 1, \ldots$}
\STATE Determine $\hat{x}_{j, k}(t+1)$ with $\phi_{k}(t)$ and $\hat{\theta}_{j, k}(t)$. 
\hfill \COMMENT {\eqref{eq:IdentModelMM}} 
\STATE Determine $\hat{e}_{n,j,k}(t)$, $j \in \mathcal{M}$  with $x_{n,k}(t)$, $\hat{x}_{n,j,k}(t)$.  
\hfill \COMMENT{\eqref{eq:IdentErrMM}}
\STATE Determine ${e}_{n,k}(t)$,  with $x_{n,k}(t)$, $x_{m,n,k}(t)$.  
\hfill \COMMENT{\eqref{eq:TrackErr}}
\STATE Compute $u_k(t)$ with $\rho_{k}(t)$, $\hat{\theta}_{j^{\ast},k}(t)$
and $e_{n,k-1}(t+1)$. 
\hfill \COMMENT {\eqref{eq:ControlMM}} 
\ENDFOR
\STATE Compute $\hat{\theta}_{j,k+1}(t)$ with $\hat{\theta}_{j, k}(t)$, $\phi_{k}(t)$ and $\hat{e}_{n,j,k}(t+1)$.
\hfill \COMMENT{\eqref{eq:ParEstMM}}
\STATE Determine $j_{k+1}^{*}$ using $\hat{e}_{n,j,k}(t+1)$.  
\hfill \COMMENT{\eqref{eq:Case1}}
\ENDFOR
\end{algorithmic}
\vspace*{0.1in}
\hrule
\label{alg:MM1}
\end{algorithm}

\begin{algorithm}[!t] 
\caption{Computational flow with Multiple Models Case $2$.}
\hrule \vspace*{0.1in}
\begin{algorithmic} 
\STATE { \bf Initialisation}: 
$\hat{\theta}_{j,0}(t) \leftarrow \mathrm{random}$, $j \in \mathcal{M}$; $j_{1}^{*}(t) = 1$, $t \in \mathcal{I}_T$
\FOR{$k = 1, 2, \ldots$}
\FOR{$t = 0, 1, \ldots$}
\STATE Determine $\hat{x}_{j, k}(t+1)$ with $\phi_{k}(t)$ and $\hat{\theta}_{j, k}(t)$. 
\hfill \COMMENT {\eqref{eq:IdentModelMM}} 
\STATE Determine $\hat{e}_{n,j,k}(t)$, $j \in \mathcal{M}$  with $x_{n,k}(t)$, $\hat{x}_{n,j,k}(t)$.  
\hfill \COMMENT{\eqref{eq:IdentErrMM}}
\STATE Determine ${e}_{n,k}(t)$,  with $x_{n,k}(t)$, $x_{m,n,k}(t)$.  
\hfill \COMMENT{\eqref{eq:TrackErr}}
\STATE Determine $j_{k}^{*}(t)$ using $\hat{e}_{n,j,k}(t)$.
\hfill \COMMENT{\eqref{eq:Case2}}
\STATE Compute $u_k(t)$ with $\rho_{k}(t)$, $\hat{\theta}_{j^{\ast},k}(t)$
and $e_{n,k-1}(t+1)$. 
\hfill \COMMENT {\eqref{eq:ControlMM}}
\ENDFOR
\STATE Compute $\hat{\theta}_{j,k+1}(t)$ with $\hat{\theta}_{j, k}(t)$, $\phi_{k}(t)$ and $\hat{e}_{n,j,k}(t+1)$.
\hfill \COMMENT{\eqref{eq:ParEstMM}}
\ENDFOR
\end{algorithmic}
\vspace*{0.1in}
\hrule
\label{alg:MM2}
\vspace{-0.1in}
\end{algorithm}

Each model $j \in \mathcal{M}$ is updated according to the following
law, similar to \eqref{eq:ParEstSM}:
\begin{align}
\hat{\theta}_{j, k+1}(t) &= \mathrm{Proj}[m] = \mathrm{Proj}\left[\hat{\theta}_{j, k}(t) + \frac{\phi_k(t)}{1 +
    \left\|\phi_k(t)\right\|^2}\hat{e}_{n, j,
    k}(t+1)\right], \label{eq:ParEstMM} 
\end{align}
where $\mathrm{Proj}[.]$ is defined in \eqref{eq:Projection}. The
algorithms \ref{alg:MM1} and \ref{alg:MM2} summarize the
above procedure for both Case $1$ and $2$.

\subsection{Convergence Analysis} \label{sec:MM-2}

We now present the primary result of this article for convergence of
Adaptive ILC using multiple models:
\begin{theorem} \label{thm:MM}
For the system \eqref{eq:Plant} with to track the
reference model \eqref{eq:RefModel}, the control law
\eqref{eq:ControlMM} along with the adaptive law \eqref{eq:ParEstMM}
guarantees the following:
\begin{enumerate}[label = { $\arabic*$. }]
\item{ $\Tilde{\theta}_j(t), \hat{\theta}_j(t) \in \ell_\infty$ for
  each $t \in \mathcal{I}_T$, for each $j \in \mathcal{M}$, i.e. the
  sequence of parametric errors $\Tilde{\theta}_{j, k}(t)$ over
  iterations --- and hence the sequence of parameter estimates
  $\hat{\theta}_{j, k}(t)$ over iterations --- is bounded for each
  sample $t$ and model $j$.} 
\item{$\hat{e}_{n, j}(t+1) \in \ell_2\cap\ell_\infty$ for each $t \in
  \mathcal{I}_T$, for each $j \in \mathcal{M}$, i.e. the sequence of
  the $n$\text{\normalfont{th}} component of the identification error
  over iterations is square-summable and bounded for each sample $t$
  and model $j$.} 
\item{With identical initial conditions on the plant and reference
  model, $\lim_{k \lra \infty} \hat{e}_{j, k}(t+1) = 0$ for each $t
  \in \mathcal{I}_T$, for each $j \in \mathcal{M}$, i.e. each
  component of the identification error vector tends to zero with
  iterations, for each sample $t$ and model $j$.} 
\item{With identical initial conditions on the plant and reference
  model, $\lim_{k \lra \infty} e_k(t+1) = 0$ for each $t \in
  \mathcal{I}_T$, i.e. each component of the tracking error vector
  tends to zero with iterations, for each sample $t$.} 
\item{$\lim_{k \lra \infty} \left\|\hat{\theta}_{j, k}(t) -
  \hat{\theta}_{j, k-p}(t)\right\|^2 = 0$ for each $t \in
  \mathcal{I}_T$, for each $j \in \mathcal{M}$, for any $p \in
  \mathbb{N}$, i.e. the parameter vector estimates converge over
  iterations for each sample $t$ and model $j$.} 
\end{enumerate}
\end{theorem}

\begin{proof}
As with the proof of Theorem \ref{thm:SM}, the proof of Theorem
\ref{thm:MM} is organized in three parts, with the statement $1$ proved in
Part $1$, statements $2$, $3$ and $4$ proved in Part $2$ and statement
$5$ proved in Part~$3$.

\vspace{0.5em} \noindent \emph{Part 1: Boundedness of Parametric Error:}

\noindent
Define a composite energy function (CEF) $V_k(t)$ as:
\begin{subequations} \label{eq:CEFMM}
\begin{align}
V_k(t) &\define \sum_{j \in \mathcal{M}} V_{j, k}(t), \\
V_{j, k}(t) &\define \Tilde{\theta}_{j, k}^{T}(t)\Tilde{\theta}_{j,
  k}(t) = \left\|\Tilde{\theta}_{j, k}(t)\right\|^2, 
\end{align}
\end{subequations}
and let $\Delta V_{j, k}(t) \define V_{j, k+1}(t) - V_{j, k}(t)$. By
arguments similar to the ones made in the proof of Theorem
\ref{thm:SM},
\begin{equation} \label{eq:DeltaVMM}
\Delta V_{j, k}(t) \leq - \left(\frac{2 +
  \left\|\phi_k(t)\right\|^2}{\left(1 +
  \left\|\phi_k(t)\right\|^2\right)^2}\right)\hat{e}_{n, j,
  k}^{2}(t+1), 
\end{equation}
or,
\begin{equation} \label{eq:DecreasingMM1}
\Delta V_{j, k}(t) \leq - \alpha_{k}^{2}(t) \hat{e}_{n, j, k}^{2}(t+1) \leq 0,
\end{equation}
where, as before, $\alpha_{k}^{2}(t)$ denotes the positive quantity
within parentheses in \eqref{eq:DeltaVMM}. Hence, $V_{j, k}(t)$ is
non-increasing for each $j$, and thus,
\begin{equation} \label{eq:DecreasingMM2}
\Delta V_{k}(t) \define V_{k+1}(t) - V_k(t) = \sum_{j \in \mathcal{M}}
\Delta V_{j, k}(t) \leq 0, 
\end{equation}
or, $V_k(t)$ is a non-increasing function. From the construction of
$V_{j, k}(t)$, it is evident that the sequence of parametric errors
$\Tilde{\theta}_{j, k}(t)$ is a bounded sequence over iterations $k$,
for each $t \in \mathcal{I}_T$, i.e. $\Tilde{\theta}_j(t) \in
\ell_\infty$. As $\theta(t)$ is bounded, we conclude that
$\hat{\theta}_j(t) \in \ell_\infty$, i.e. the sequence of parameter
estimates $\hat{\theta}_{j, k}(t)$ is bounded over iterations $k$, for
each $t \in \mathcal{I}_T$.

\vspace{0.5em} \noindent \emph{Part 2: Convergence of Errors:}

\noindent
From \eqref{eq:DecreasingMM1}, $\lim_{N \lra \infty} \left|V_{j,
  N+1}(t) - V_{j, 1}(t)\right| < \infty$ for each $t \in
\mathcal{I}_T$. Rewriting this,
$$
\lim_{N \lra \infty} \left|\sum_{k = 1}^{N} \Delta V_{j, k}(t)\right|
\leq \lim_{N \lra \infty} \sum_{k = 1}^{N}
\alpha_{k}^{2}(t)\hat{e}_{n, j, k}^{2}(t+1) < \infty. 
$$
Using the same arguments as earlier, $\hat{e}_{n, j}(t+1) \in \ell_2
\cap \ell_\infty$, and
\begin{equation} \label{eq:ConvIdentErrMM}
\lim_{k \lra \infty} \hat{e}_{n, j, k}(t+1) = 0
\end{equation}
for each $t \in \mathcal{I}_T$, for each model $j \in
\mathcal{M}$. Under the assumption of identical initial conditions,
this implies that:
\begin{equation} \label{eq:ConvIdentErrVecMM}
\lim_{k \lra \infty} \hat{e}_{j, k}(t+1) = 0.
\end{equation}

Now consider eq. \eqref{eq:MMError}. While we know that $\hat{e}_{n,
  j, k}(t+1) \lra 0$ as $k \lra \infty$, the actual sequence
$\hat{e}_{n, j^*, k}(t+1)$ that acts as a forcing function here
depends on the switching criterion considered, either \eqref{eq:Case1}
or \eqref{eq:Case2}. We now show that $\hat{e}_{n, j^*, k}(t+1) \lra
0$ as $k \lra \infty$, where $j^*$ denotes $j_{k}^{*}$ or
$j_{k}^{*}(t)$, depending on whether criterion \eqref{eq:Case1} or
\eqref{eq:Case2} is used. For notational simplicity, let $\hat{e}_{j,
  k}$ denote $\hat{e}_{n, j, k}(t+1)$. Construct the following
sequence at each sample $t$:
\begin{equation} \label{eq:Sequence}
S = \underbrace{\hat{e}_{1, 1}, \ldots, \hat{e}_{M, 1}}_{j \in
  \mathcal{M}, k = 1}, \underbrace{\hat{e}_{1, 2}, \ldots, \hat{e}_{M, 2}}_{j 
  \in \mathcal{M}, k = 2}, \ldots, \underbrace{\hat{e}_{1, k}, \ldots,
  \hat{e}_{M, k}}_{j \in \mathcal{M}, k = k}, \ldots 
\end{equation}
This is a sequence of identification errors of each model, considered
one iteration after another. Since $\hat{e}_{j, k} \lra 0$ as $k \lra
\infty$, the above sequence $S \lra 0$ as $k \lra \infty$. Then, if
$S^*$ denotes any subsequence of $S$, $S^* \lra 0$ as $k \lra
\infty$. We exploit this fact to show that using either criterion
\eqref{eq:Case1} or \eqref{eq:Case2}, the forcing function
$\hat{e}_{n, j^*, k}(t+1)$ in \eqref{eq:MMError} converges to $0$ as
$k \lra \infty$.

With criterion \eqref{eq:Case1}, switching takes place only once every
iteration, at the end. The optimal model $j_{k}^{*}$ does not depend
on the sample $t$. Then, the forcing function sequence in
\eqref{eq:MMError} can be written as $S^* = \hat{e}_{j_{1}^{*}, 1},
\hat{e}_{j_{2}^{*}, 2}, \ldots$, with each $j_{k}^{*} \in
\mathcal{M}$. $S^*$ is evidently a subsequence of $S$, and as $S \lra
0$ as $k \lra \infty$, $S^* \lra 0$ as $k \lra \infty$, and hence
$\hat{e}_{n, j^*, k}(t+1) \lra 0$ as $k \lra \infty$, for each $t$.

The arguments for criterion \eqref{eq:Case2} are very similar. The
optimal model $j_{k}^{*}(t)$ is now dependent on the sample $t$. For a
given $t$, the forcing function sequence in \eqref{eq:MMError} can be
written as $S^*(t) = \hat{e}_{j_{1}^{*}(t), 1}, \hat{e}_{j_{2}^{*}(t),
  2}, \ldots$, with each $j_{k}^{*}(t) \in \mathcal{M}$. $S^*(t)$ is a
subsequence of $S$ for each $t$, and by the above arguments, $S^*(t)
\lra 0$ as $k \lra \infty$, and hence $\hat{e}_{n, j^*, k}(t+1) \lra
0$ as $k \lra \infty$, for each $t$.

The minor difference between the two arguments lies in the fact that
the subsequence constructed depends on the sample $t$ in the second
case. For both criteria \eqref{eq:Case1} and \eqref{eq:Case2},
\begin{equation} \label{eq:ForcFunc}
\lim_{k \lra \infty} \hat{e}_{n, j^*, k}(t+1) = 0
\end{equation}
for each $t \in \mathcal{I}_T$. Then, eq. \eqref{eq:MMError} is an
iteration-domain difference equation with a forcing function that
converges to $0$. As $0 < \beta < 1$, it is evident that:
\begin{equation} \label{eq:ConvTrackErrMM}
\lim_{k \lra \infty} e_{n, k}(t+1) = 0
\end{equation}
for each $t \in \mathcal{I}_T$. Under the assumption of identical
initial conditions,
\begin{equation} \label{eq:ConvTrackErrVecMM}
\lim_{k \lra \infty} e_k(t+1) = 0.
\end{equation}

\vspace{0.5em} \noindent \emph{Part 3: Convergence of Parameter Estimates:}

\noindent
The final part of the proof is straightforward and simply extends the
arguments made in the corresponding part of the proof of Theorem
\ref{thm:SM} to the multiple model case. It can first be shown that
the relation $\left\|\mathrm{Proj}[m] - \hat{\theta}_{j, k}(t)\right\|
\leq \left\|m - \hat{\theta}_{j, k}(t)\right\|$ always holds, for each
model $j \in \mathcal{M}$. Then, using \eqref{eq:ParEstMM},
$$
\left\|\hat{\theta}_{j, k+1}(t) - \hat{\theta}_{j, k}(t)\right\|^2
\leq \hat{e}_{n, j, k}^{2}(t+1). 
$$
Then, summing the above inequality over iterations $k$ and using
the properties of the $\ell_2$ sequence $\hat{e}_{n, j}(t+1)$,
\begin{equation} \label{eq:ConvParEst1MM}
\lim_{k \lra \infty} \left\|\hat{\theta}_{j, k+1}(t) -
\hat{\theta}_{j, k}(t)\right\| = 0 
\end{equation}
for each $t \in \mathcal{I}_T$. Thus, for each model $j \in
\mathcal{M}$, parameter estimates one iteration apart converge for
each sample $t$. For parameter estimates $p$ iterations apart,
$\left\|\hat{\theta}_{j, k}(t) - \hat{\theta}_{j, k-p}(t)\right\|$ is
written as a telescoping series, as shown earlier. By the same
arguments,
\begin{equation} \label{eq:ConvParEstPMM}
\lim_{k \lra \infty} \left\|\hat{\theta}_{j, k}(t) - \hat{\theta}_{j,
  k-p}(t)\right\| = 0 
\end{equation}
for each $t \in \mathcal{I}_T$, for each model $j \in
\mathcal{M}$. This concludes the proof of Theorem \ref{thm:MM}.
\end{proof}

Summarizing the results of this section, we have presented an approach
using multiple estimation models to solve the discrete-time Adaptive
ILC problem. This approach is enabled by using each model's
identification error in updating the corresponding parameter
estimates. The control law is formulated based on the optimal model at
sample $t$, in iteration $k$. We provide two options for switching
between models --- either once in an iteration or once every sample
--- and each option has its own criterion. Using either criterion, we 
prove that each component of the identification and tracking error vectors 
converge to
$0$ with iterations $k$. A key step in this proof is to show that the
sequence of identification errors corresponding to the best model ---
$\hat{e}_{n, j^*, k}(t+1)$ --- converges to $0$ as $k \lra \infty$,
using either criterion. As with the strategy in Section
\ref{sec:SM-AILC}, the proof of convergence does not involve the KTL and the properties of $V_k(t)$ and $\ell_2$ sequences
are used instead.

\section{Simulation Examples} \label{sec:Simulations}

In this section, we present simulation examples to demonstrate the
efficacy of the single-model strategy proposed in Section
\ref{sec:SM-AILC}, and the two switching strategies with multiple
models in Section \ref{sec:MM-AILC}. Four different first-order
systems are considered: a linear, time-invariant system not subjected
to disturbances (LTI), a linear, time-varying system subjected to
disturbances (LTV-D), a nonlinear system not subjected to disturbances
(NL) and a nonlinear system subjected to disturbances (NL-D). In each
example, the time interval for each iteration is $\left\{0, 1, \ldots,
100\right\}$, and hence the time index $t$ is in the set
$\mathcal{I}_T = \left\{0, 1, \ldots, 99\right\}$. The parameter
$\beta$ in the control laws \eqref{eq:ControlSM} and
\eqref{eq:ControlMM} is set to $0.2$. Zero initial conditions on the
plant and reference are assumed in all examples. For the multiple-model 
cases, the number of models is set to $M = 10$, and parameters
are initialized randomly in the parameter space. The strategy that
applies the single model control law \eqref{eq:ControlSM} is
designated ``SM'', and the strategies that use the multiple-model
control law \eqref{eq:ControlMM} with criterion \eqref{eq:Case1} or
\eqref{eq:Case2} are designated ``MM - Case $1$'' and ``MM - Case
$2$'' respectively. For each example, the objective is to track the
following iteration-invariant reference:
\begin{equation} \label{eq:RefTraj}
x_m(t) = \pi^2\left(2 - 3\sin^3(2\pi t/100)\right)\sin(2\pi t/100)/10,
\end{equation}
which is similar to the trajectory considered in
\cite{CHX2008}. The efficacy of each strategy is measured based on the
peak identification and tracking errors over iterations, which ideally
converge to zero. This is the same as considering the $\infty$-norm of
both errors, defined below for the identification error:
\begin{equation} \label{eq:InftyNorm}
\left\|\hat{\mathbf{e}}_k\right\|_{\infty} = \max_t \left|x_k(t+1) -
\hat{x}_k(t+1)\right| = \max_t \left|\hat{e}_k(t+1)\right|, 
\end{equation}
and defined similarly for the tracking error. An iteration-invariant
trajectory is considered for simplicity, to highlight the advantages
of faster convergence in multiple models. The final example in this
section presents results for tracking an iteration-varying trajectory.

\begin{figure}[!t]
\centering

\begin{subfigure}{0.49\textwidth}
	\includegraphics[width = \textwidth]{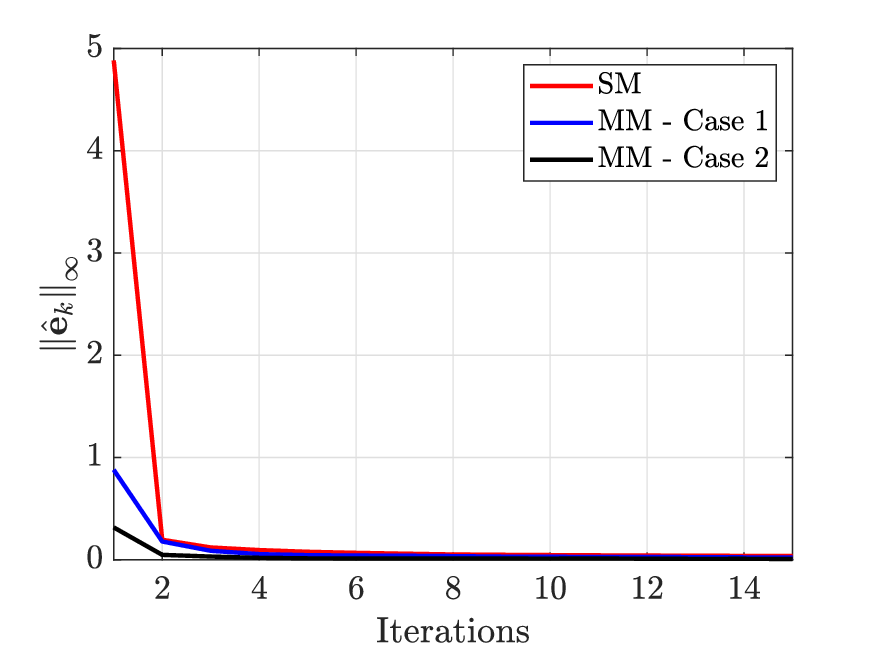}
	\caption{Maximum Identification Error over iterations.}
\end{subfigure}
\hfill
\begin{subfigure}{0.49\textwidth}
	\includegraphics[width = \textwidth]{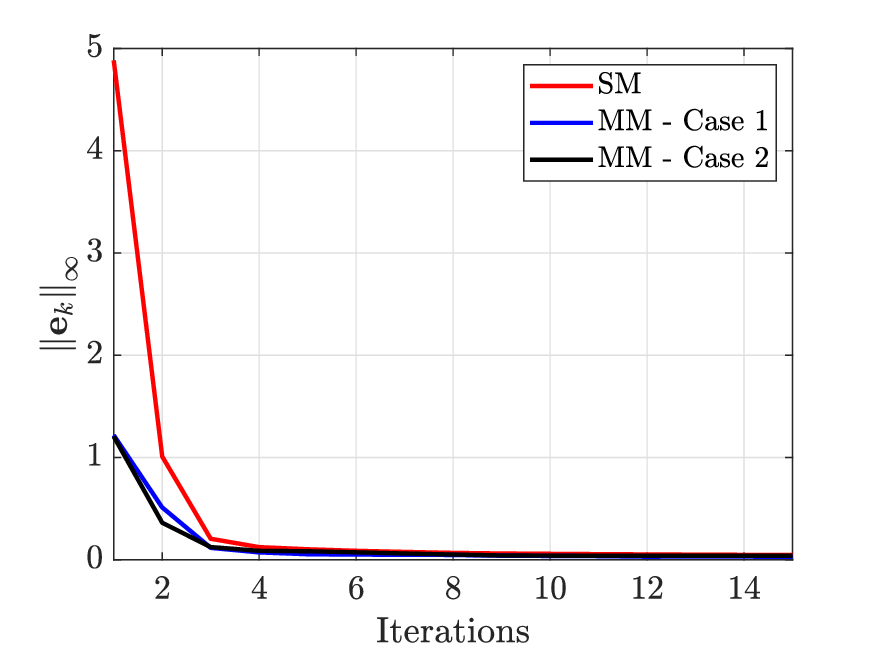}
	\caption{Maximum Tracking Error over iterations.}
\end{subfigure}

\caption{Example 1: Error profiles over iterations for an LTI system without disturbances.}
\label{fig:LTI}
\hrulefill
\vspace{-0.3cm}

\end{figure}

\subsection{Example 1: LTI System without Disturbances} \label{sec:LTI}
Consider the system:
\begin{equation} \label{eq:LTI}
	x_k(t+1) = 0.5x_k(t) + u_k(t),
\end{equation}
a simple, stable LTI system without disturbances. The objective is for
$x_k(t)$ to track the reference $x_m(t)$ in \eqref{eq:RefTraj}. The
results for identification and tracking performance are shown in
Fig.~\ref{fig:LTI}, in terms of the peak amplitude of errors over
iterations for each strategy. It is evident that both multiple-model
strategies converge faster than the single-model strategy, mainly
because the initialization of multiple estimation models leads to
better estimates in earlier iterations, hence improving transient
performance. Further, the multiple model strategy with criterion
\eqref{eq:Case2}, i.e. MM - Case $2$ converges marginally faster than
MM - Case $1$, due to models switching more frequently.

\begin{figure}[!t]
\centering

\begin{subfigure}{0.49\textwidth}
	\includegraphics[width = \textwidth]{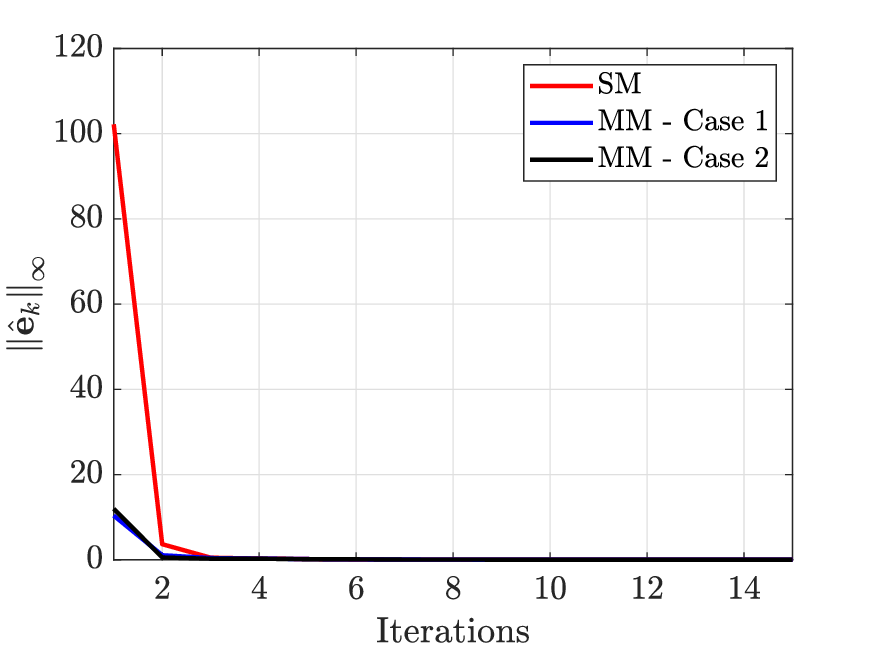}
	\caption{Maximum Identification Error over iterations.}
\end{subfigure}
\hfill
\begin{subfigure}{0.49\textwidth}
	\includegraphics[width = \textwidth]{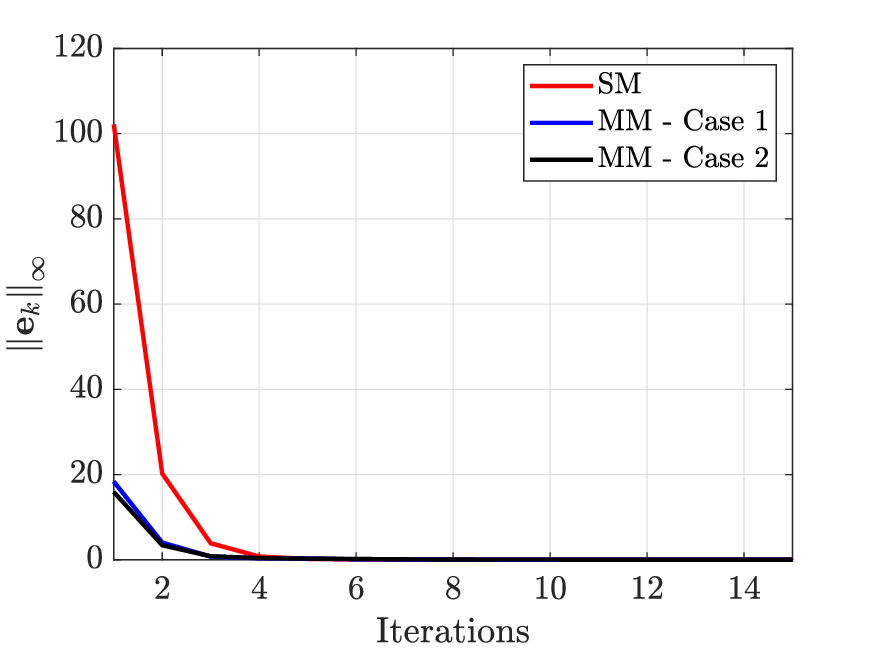}
	\caption{Maximum Tracking Error over iterations.}
\end{subfigure}

\caption{Example 2: Error profiles over iterations for an LTV system with disturbances.}
\label{fig:LTV-D}
\hrulefill
\vspace{-0.3cm}

\end{figure}

\begin{figure}[!t]
\centering

\begin{subfigure}{0.49\textwidth}
	\includegraphics[width = \textwidth]{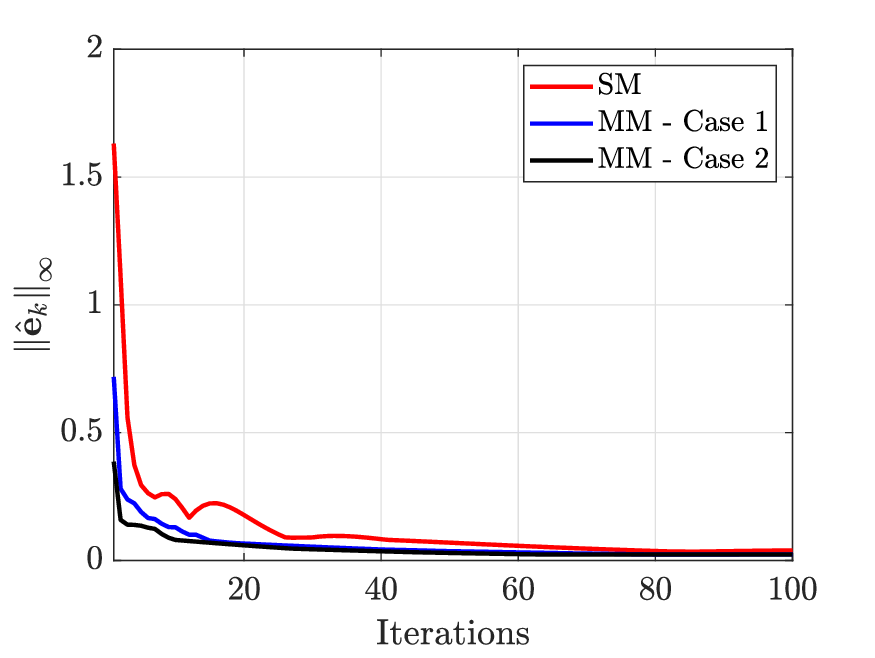}
	\caption{Maximum Identification Error over iterations.}
\end{subfigure}
\hfill
\begin{subfigure}{0.49\textwidth}
	\includegraphics[width = \textwidth]{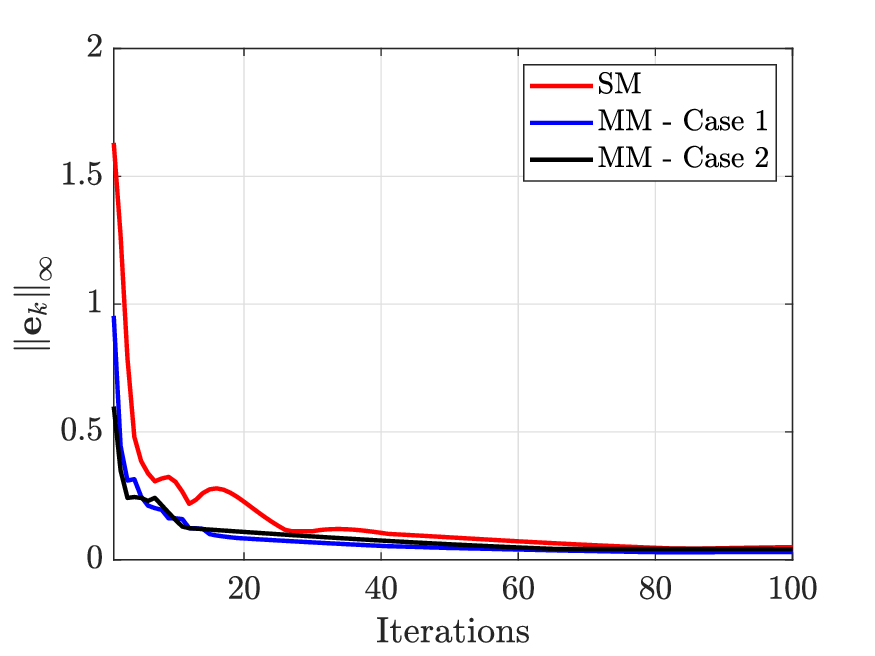}
	\caption{Maximum Tracking Error over iterations.}
\end{subfigure}

\caption{Example 3: Error profiles over iterations for a nonlinear system without
  disturbances.} 
\label{fig:NL}
\hrulefill
\vspace{-0.3cm}

\end{figure}

\begin{figure}[!t]
\centering

\begin{subfigure}{0.49\textwidth}
	\includegraphics[width = \textwidth]{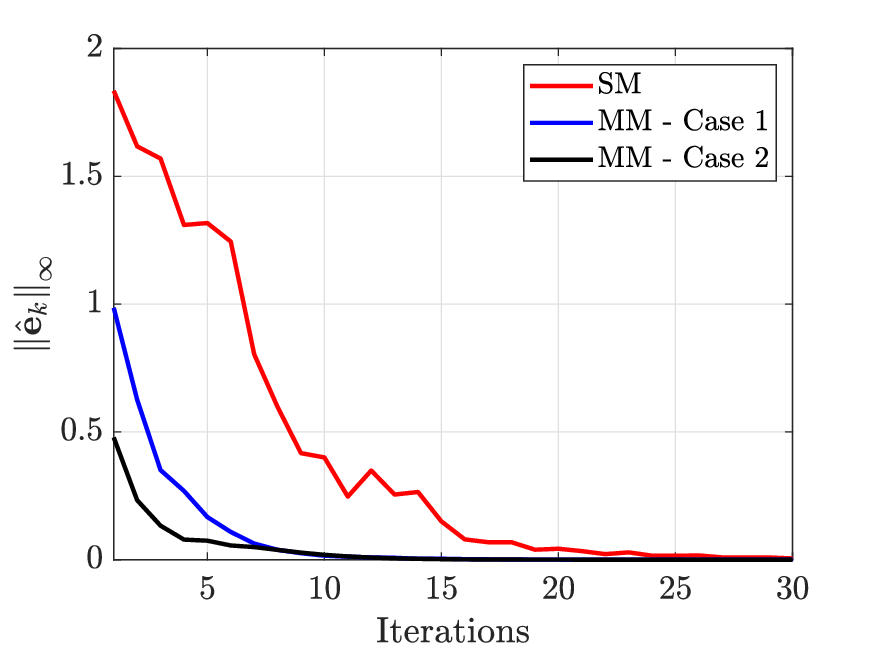}
	\caption{Maximum Identification Error over iterations.}
\end{subfigure}
\hfill
\begin{subfigure}{0.49\textwidth}
	\includegraphics[width = \textwidth]{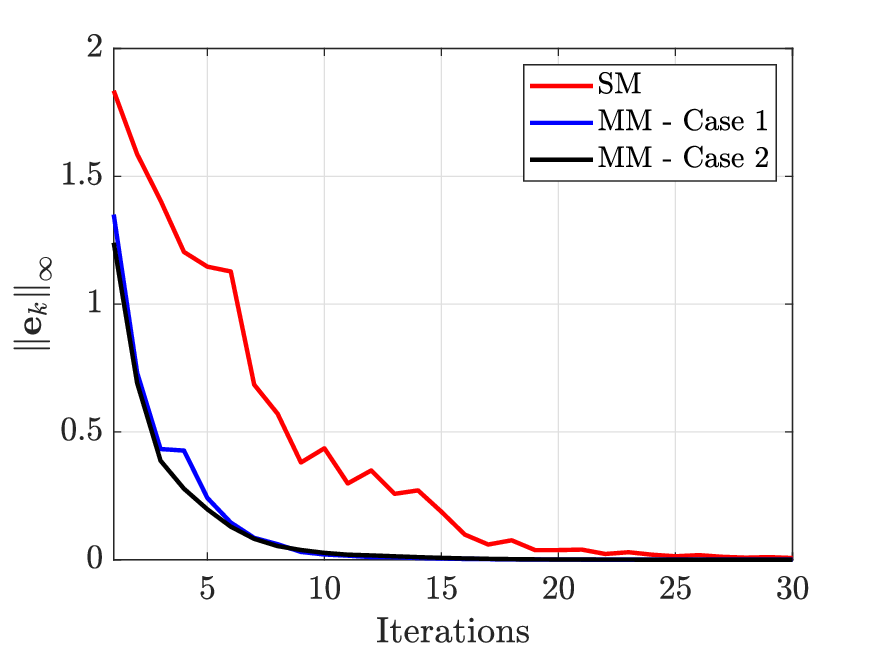}
	\caption{Maximum Tracking Error over iterations.}
\end{subfigure}

\caption{Example 4: Error profiles over iterations for a nonlinear system with disturbances.}
\label{fig:NL-D}
\hrulefill
\vspace{-0.3cm}

\end{figure}

\subsection{Example 2: LTV System with Disturbances} \label{sec:LTV-D}
In this example, the following LTV system with disturbances is considered:
\begin{equation} \label{eq:LTV-D}
	x_k(t+1) = \theta_1(t)x_k(t) + b(t)u_k(t) + d(t),
\end{equation}
where $\theta_1(t) = 1 + 0.5\sin(t)$, $b(t) = 3 + 0.5\sin(2\pi t)$ and
$d(t) = \sin^3(2\pi t)$, an external disturbance. The results for
achieving the tracking objective are shown in Fig.~\ref{fig:LTV-D}. As
before, the two multiple-model cases achieve faster convergence due to
improved transient response, whereas the single-model case has very
poor transient response due to large initial parametric errors. This
example also demonstrates the first instance of time-varying
parameters being successfully identified over iterations.

\subsection{Example 3: Nonlinear System without Disturbances} \label{sec:NL}
The following nonlinear system is considered:
\begin{equation} \label{eq:NL}
	x_k(t+1) = \theta_1(t)\sin^2\left(x_k(t)\right) + b(t)u_k(t),
\end{equation}
where $\theta_1(t) = 1 + 0.5\sin(t)$ and $b(t) = 3 + 0.5\sin(2\pi t)$,
as before. Note the nonlinearity in the regression vector, in contrast
to the previous examples. Fig.~\ref{fig:NL} shows the results for
identification and tracking performance for this system, in terms of
the peak amplitude of the errors. The two multiple-model cases are
seen to achieve faster convergence, and in particular, the strategy
``MM - Case $2$'' converges marginally faster due to higher frequency
of switching.

\subsection{Example 4: Nonlinear System with Disturbances} \label{sec:NL-D}
Finally, the most general system is considered:
\begin{equation} \label{eq:NL-D}
	x_k(t+1) = \theta_1(t)\sin^2\left(x_k(t)\right) + b(t)u_k(t) + d(t),
\end{equation}
where $\theta_1(t) = 1.5 + 0.5\sin(t)$, $b(t) = 3 + 0.5\sin(2\pi t)$
and $d(t) = \sin^3(2\pi t)$, the same disturbance considered in
\eqref{eq:LTV-D}. This is very similar to the system considered in
\cite{CHX2008}. Fig.~\ref{fig:NL-D} shows the performance for tracking
the reference \eqref{eq:RefTraj} over iterations, in terms of peak
identification and tracking errors. It is evident that the convergence
for both multiple-model cases is significantly faster than the 
single-model case, with MM - Case $2$ providing the fastest
convergence. Interestingly, the errors for this system also converge
faster than the errors for the system \eqref{eq:NL}, which was not
affected by disturbances. This is due to the presence of $d(t)$ and
its estimate, resulting in a persistently exciting control law.

\begin{figure}[!t]
\centering

\begin{subfigure}{0.49\textwidth}
	\includegraphics[width = \textwidth]{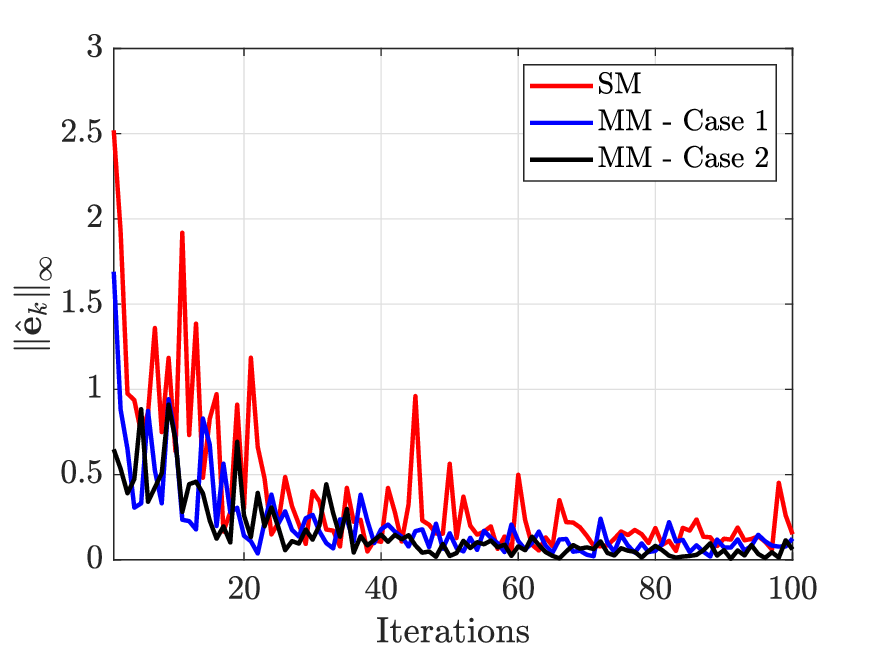}
	\caption{Maximum Identification Error over iterations.}
\end{subfigure}
\hfill
\begin{subfigure}{0.49\textwidth}
	\includegraphics[width = \textwidth]{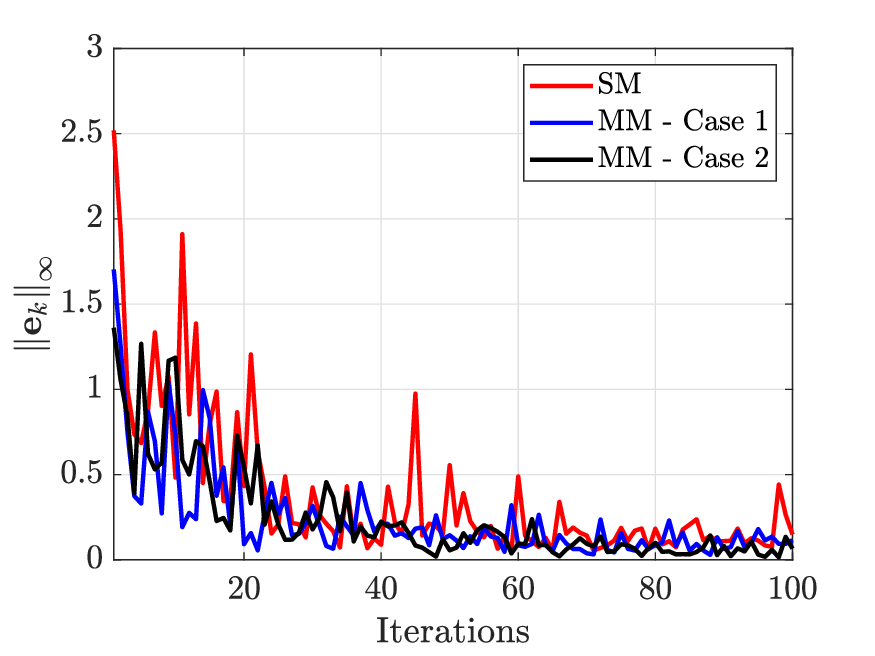}
	\caption{Maximum Tracking Error over iterations.}
\end{subfigure}

\caption{Example 5: Error profiles over iterations for a nonlinear system with
  disturbances, with an iteration-varying reference trajectory.} 
\label{fig:NL-D-IV}
\hrulefill

\end{figure}

\subsection{Example 5: Iteration-varying Reference Trajectory}
We also present results for the system \eqref{eq:NL-D} tracking an
iteration-varying reference trajectory:
\begin{equation} \label{eq:RefTrajIV}
x_{m, k}(t) = \vartheta(k)\pi^2\left(2 - 3\sin^3(2\pi
t/100)\right)\sin(2\pi t/100)/10, 
\end{equation}
where $\vartheta(k) \sim \mathcal{U}[-0.5, 0.5]$, i.e. a uniformly
distributed random variable between $-0.5$ and $0.5$, in each
iteration $k$. This is similar to the trajectory considered in \cite{CHX2008}.
The results for this example are shown in
Fig.~\ref{fig:NL-D-IV}. It is easily seen that all errors are
decreasing and are quite close to $0$. Further, the two 
multiple-model strategies are seen to perform better than the single-model
strategy. To reinforce this, Table~\ref{tab:Metric} presents the 
root-mean-square value of peak error amplitudes over iterations for
identification and tracking errors for all three strategies. This
metric is defined below for the identification error:
\begin{equation} \label{eq:Metric}
\mathrm{Metric} = \sqrt{\frac{1}{K}\sum_{k = 1}^{K}
  \left\|\hat{\mathbf{e}}_k\right\|^{2}_{\infty}}, 
\end{equation}
and defined similarly for the tracking error. $K$ denotes the total
number of iterations in the simulation. Lower values of this metric
indicate better performance. From Table~\ref{tab:Metric}, it is
evident that both multiple-model strategies have significantly smaller
values of this metric, indicating smaller identification and tracking
errors and faster convergence. The smallest values are in MM - Case
$2$, indicating that this strategy achieves the best possible
performance. Further, Table~\ref{tab:Metric2} presents
the number of iterations $k^*$ taken for tracking convergence.
In particular, we consider the number of iterations taken for the peak
tracking error $\left\|\mathbf{e}_k\right\|_{\infty}$ to fall below
$2\%$ of the first iteration peak tracking error
$\left\|\mathbf{e}_1\right\|_{\infty}$ of the single model strategy. It is
once again evident that both multiple model strategies converge faster, with
MM - Case $2$ converging fastest. As mentioned earlier, this is partly because
the initialization of multiple estimation models leads to better estimates
in earlier iterations, leading to faster convergence.

In conclusion, all simulation examples demonstrate that the proposed
strategies result in convergence of identification and tracking errors
to zero. The multiple model strategies converge much faster than the
single model strategy, and the multiple model strategy with criterion
\eqref{eq:Case2} converges faster than that with criterion
\eqref{eq:Case1}. Note how the initial error magnitudes were larger
for the linear systems compared to the nonlinear systems. This is due
to the presence of the nonlinearity $\sin^2\left(x_k(t)\right)$, which
is always bounded between $0$ and $1$ irrespective of the value of
$x_k(t)$. The nonlinear system subjected to disturbances also shows
faster error convergence than the system without disturbances.

\begin{table}
\centering

\caption{Root Mean-Square Maximum Errors}
\begin{tabular}{lcc} \toprule
Strategy & Identification Error & Tracking Error \\ \midrule
Single Model & $6.3288$ & $6.2884$ \\
Multiple Models -- Case $1$ & $3.0544$ & $3.4948$ \\
Multiple Models -- Case $2$ & $2.0728$ & $3.1338$ \\ \bottomrule
\end{tabular}

\label{tab:Metric}
\vspace{0.2cm}

\end{table}

\begin{table}
\centering

\caption{Iterations for Tracking Convergence}
\begin{tabular}{lc} \toprule
Strategy & Number of Iterations $k^*$ \\ \midrule
Single Model & $79$ \\
Multiple Models -- Case $1$ & $58$ \\
Multiple Models -- Case $2$ & $47$ \\ \bottomrule
\end{tabular}

\label{tab:Metric2}
\vspace{0.4cm}
\hrulefill
\vspace{-0.3cm}

\end{table}

\section{Concluding Remarks} \label{sec:Conclusion}
In this article, we have proposed a complete framework for using the
Multiple Models, Switching and Tuning (MMST) methodology in the
context of discrete-time Adaptive Iterative Learning Control
(ILC). First, the single estimation model case is considered, a
new control and identification scheme is presented, and convergence is
proved using the properties of square-summable, or $\ell_2$
sequences. The update law for parameter estimates uses the
identification error rather than the tracking error, in contrast to
existing Adaptive ILC schemes. This enables the extension to multiple
estimation models. In the case of multiple models, we have described two
criteria for switching between models --- either at the end of each
iteration or at each sample. In both options, convergence is proved
in a unified manner using the properties of square-summable
sequences. An extensive set of simulation results are presented for
four different types of systems. In all cases, it is seen that the
identification and tracking errors converge to zero. In particular, it
is seen that the second switching criterion for multiple-models
outperforms the first, which in turn outperforms the single model
case.

A drawback of the strategies presented here is their high
computational complexity, particularly for the second switching
criterion \eqref{eq:Case2} in multiple models. It is known that the
Multiple Models with Second-Level Adaptation (MM-SLA) scheme
\cite{NH2011, MRG2018} has lower computational complexity compared to
MMST, as a much smaller number of estimation models, is required. This
was explored for Adaptive ILC in a contraction-mapping setting in
\cite{PBHMG21b}, and an interesting avenue for future work is to
explore this in the context of CEF-based Adaptive ILC. Further, as
mentioned in Section \ref{sec:Introduction}, the techniques proposed
here can be extended to the case with time- and iteration-varying
parameters, and also to the case when the control direction is
unknown.


\begin{thebibliography}{10}
\providecommand{\url}[1]{#1}
\csname url@samestyle\endcsname
\providecommand{\newblock}{\relax}
\providecommand{\bibinfo}[2]{#2}
\providecommand{\BIBentrySTDinterwordspacing}{\spaceskip=0pt\relax}
\providecommand{\BIBentryALTinterwordstretchfactor}{4}
\providecommand{\BIBentryALTinterwordspacing}{\spaceskip=\fontdimen2\font plus
\BIBentryALTinterwordstretchfactor\fontdimen3\font minus
  \fontdimen4\font\relax}
\providecommand{\BIBforeignlanguage}[2]{{%
\expandafter\ifx\csname l@#1\endcsname\relax
\typeout{** WARNING: IEEEtran.bst: No hyphenation pattern has been}%
\typeout{** loaded for the language `#1'. Using the pattern for}%
\typeout{** the default language instead.}%
\else
\language=\csname l@#1\endcsname
\fi
#2}}
\providecommand{\BIBdecl}{\relax}
\BIBdecl

\small{
\bibitem{Xu2011}
J.-X. Xu, ``A survey on iterative learning control for nonlinear systems,''
  \emph{International Journal of Control}, vol.~84, no.~7, pp. 1275--1294,
  2011.

\bibitem{ACM2007}
H.-S. Ahn, Y.~Chen, and K.~L. Moore, ``Iterative learning control: {Brief}
  survey and categorization,'' \emph{IEEE Transactions on Systems, Man, and
  Cybernetics -- Part C: Applications and Reviews}, vol.~37, no.~6, pp.
  1099--1121, Nov. 2007.

\bibitem{KLM93}
K.~L. Moore, \emph{Iterative Learning Control for Deterministic Systems}.\hskip
  1em plus 0.5em minus 0.4em\relax London, United Kingdom: Springer-Verlag,
  1993.

\bibitem{BTA2006}
D.~A. Bristow, M.~Tharayil, and A.~G. Alleyne, ``A survey of iterative learning
  control,'' \emph{IEEE Control Systems Magazine}, vol.~26, no.~3, pp. 96--114,
  Jun. 2006.

\bibitem{BX1998}
Z.~Bien and J.-X. Xu, Eds., \emph{Iterative Learning Control: Analysis, Design,
  Integration and Applications}.\hskip 1em plus 0.5em minus 0.4em\relax New
  York, NY, USA: Springer Science + Business Media, 1998.

\bibitem{AKM84}
S.~Arimoto, S.~Kawamura, and F.~Miyazaki, ``Bettering operation of robots by
  learning,'' \emph{Journal of Robotic Systems}, vol.~1, no.~2, pp. 123--140,
  1984.

\bibitem{CM2002}
Y.~Chen and K.~L. Moore, ``{An optimal design of PD-type iterative learning
  control with monotonic convergence},'' in \emph{Proceedings of the IEEE
  International Symposium on Intelligent Control}, Vancouver, Canada, Oct.
  2002, pp. 55--60.

\bibitem{CWS97}
Y.~Chen, C.~Wen, and M.~Sun, ``A robust high-order {P-type} iterative learning
  controller using current iteration tracking error,'' \emph{International
  Journal of Control}, vol.~68, no.~2, pp. 331--342, 1997.

\bibitem{CL1996}
C.-J. Chien and J.-S. Liu, ``{A P-type iterative learning controller for robust
  output tracking of nonlinear time-varying systems},'' \emph{International
  Journal of Control}, vol.~64, no.~2, pp. 319--334, 1996.

\bibitem{WP18}
Z.~Wang, C.~P. Pannier, K.~Barton, and D.~J. Hoelzle, ``Application of robust
  monotonically convergent spatial iterative learning control to microscale
  additive manufacturing,'' \emph{Mechatronics}, vol.~56, pp. 157--165, Dec.
  2018.

\bibitem{AA21}
A.~A. Armstrong and A.~G. Alleyne, ``A multi-input single-output iterative
  learning control for improved material placement in extrusion-based additive
  manufacturing,'' \emph{Control Engineering Practice}, vol. 111, pp. 104\,783:
  1--11, Jun. 2021.

\bibitem{LO2020}
A.~Laracy and H.~Ossareh, ``Constraint management for batch processes using
  iterative learning control and reference governors,'' in \emph{Proceedings of
  the 2nd Conference on Learning for Dynamics and Control}, Jun. 2020, pp.
  340--349.

\bibitem{LMKL21}
X.~Liu, L.~Ma, X.~Kong, and K.~Lee, ``Robust model predictive iterative
  learning control for iteration-varying-reference batch processes,''
  \emph{IEEE Transactions on Systems, Man and Cybernetics: Systems}, vol.~51,
  no.~7, pp. 4238--4250, Jul. 2021.

\bibitem{CWOH2011}
Z.~H. Chen, Y.~Wang, P.~Ouyang, J.~Huang, and W.~J. Zhang, ``A novel
  iteration-based controller for hybrid machine systems for trajectory tracking
  at the end-effector level,'' \emph{Robotica}, vol.~29, no.~2, pp. 317--324,
  March 2011.

\bibitem{OZG2006}
P.~R. Ouyang, W.~J. Zhang, and M.~M. Gupta, ``An adaptive switching learning
  control method for trajectory tracking of robotic manipulators,''
  \emph{Mechatronics}, vol.~16, no.~1, pp. 51--61, February 2006.

\bibitem{CWC2020}
S.~Chen, Z.~Wang, A.~Chakraborty, M.~Klecka, G.~Saunders, and J.~Wen, ``Robotic
  deep rolling with iterative learning motion and force control,'' \emph{IEEE
  Robotics and Automation Letters}, vol.~5, no.~4, pp. 5581--5588, Oct. 2020.

\bibitem{CC2020}
Y.~Chen, B.~Chu, C.~T. Freeman, and Y.~Liu, ``{Generalized iterative learning
  control with mixed system constraints: A gantry robot based verification},''
  \emph{Control Engineering Practice}, vol.~95, pp. 104\,260: 1--11, Feb. 2020.

\bibitem{MA2021}
R.~Mengacci, F.~Angelini, M.~G. Catalano, G.~Grioli, A.~Bicchi, and
  M.~Garabini, ``On the motion/stiffness decoupling property of articulated
  soft robots with application to model-free torque iterative learning
  control,'' \emph{The International Journal of Robotics Research}, vol.~40,
  no.~1, pp. 348--374, Jan. 2021.

\bibitem{NG2002}
M.~Norrl\"of and S.~Gunnarsson, ``Time and frequency domain convergence
  properties in iterative learning control,'' \emph{International Journal of
  Control}, vol.~75, no.~14, pp. 1114--1126, 2002.

\bibitem{AMC2007}
H.-S. Ahn, K.~L. Moore, and Y.~Chen, \emph{Iterative Learning Control:
  Robustness and Monotonic Convergence for Interval Systems}.\hskip 1em plus
  0.5em minus 0.4em\relax London, United Kingdom: Springer-Verlag, 2007.

\bibitem{DW1998}
D.~Wang, ``Convergence and robustness of discrete time nonlinear systems with
  iterative learning control,'' \emph{Automatica}, vol.~34, no.~11, pp.
  1445--1448, Nov. 1998.

\bibitem{LB1996}
H.-S. Lee and Z.~Bien, ``Study on robustness of iterative learning control with
  non-zero initial error,'' \emph{International Journal of Control}, vol.~64,
  no.~3, pp. 345--349, 1996.

\bibitem{SBT2022}
Z.~Shahriari, B.~Bernhardsson, and O.~Troeng, ``Convergence analysis of
  iterative learning control using pseudospectra,'' \emph{International Journal
  of Control}, vol.~95, no.~1, pp. 269--281, 2022.

\bibitem{XT2001}
J.-X. Xu and Y.~Tan, ``A suboptimal learning control scheme for non-linear
  systems with time-varying parametric uncertainties,'' \emph{Optimal Control
  Applications \& and Methods}, vol.~22, no.~3, pp. 111--126, May/Jun. 2001.

\bibitem{XT2002}
------, ``A composite energy function-based learning control approach for
  nonlinear systems with time-varying parametric uncertainties,'' \emph{IEEE
  Transactions on Automatic Control}, vol.~47, no.~11, pp. 1940--1945, Nov.
  2002.

\bibitem{FR2000}
M.~French and E.~Rogers, ``Non-linear iterative learning by an adaptive
  {Lyapunov} technique,'' \emph{International Journal of Control}, vol.~73,
  no.~10, pp. 840--850, 2000.

\bibitem{Xu2002}
J.-X. Xu, ``The frontiers of iterative learning control - {II},''
  \emph{Systems, Control and Information}, vol.~46, no.~5, pp. 233--243, 2002.

\bibitem{XTL2003}
J.-X. Xu, Y.~Tan, and T.-H. Lee, ``Iterative learning control design based on
  composite energy function with input saturation,'' in \emph{Proceedings of
  the 2003 American Control Conference}, Denver, CO, USA, Jun. 2003, pp.
  5129--5134.

\bibitem{TC2007}
A.~Tayebi and C.-J. Chien, ``A unified adaptive iterative learning control
  framework for uncertain nonlinear systems,'' \emph{IEEE Transactions on
  Automatic Control}, vol.~52, no.~10, pp. 1907--1913, Oct. 2007.

\bibitem{AT2003}
A.~Tayebi, ``Adaptive iterative learning control for robot manipulators,'' in
  \emph{Proceedings of the 2003 American Control Conference}, Denver, CO, USA,
  Jun. 2003, pp. 4518--4523.

\bibitem{AT2004}
------, ``Adaptive iterative learning control for robot manipulators,''
  \emph{Automatica}, vol.~40, no.~7, pp. 1195--1203, Jul. 2004.

\bibitem{LSWT19}
R.~Lee, L.~Sun, Z.~Wang, and M.~Tomizuka, ``Adaptive iterative learning control
  of robot manipulators for friction compensation,'' in \emph{8th IFAC
  Symposium on Mechatronic Systems (MECHATRONICS) 2019}, Vienna, Austria, Sep.
  2019, pp. 175--180.

\bibitem{WYC19}
L.~Wu, Q.~Yan, and J.~Cai, ``Neural network-based adaptive learning control for
  robot manipulators with arbitrary initial errors,'' \emph{IEEE Access},
  vol.~7, pp. 180\,194--180\,204, Dec. 2019.

\bibitem{HCMS21}
D.~Huang, Y.~Chen, D.~Meng, and P.~Sun, ``Adaptive iterative learning control
  for high-speed train: {A} multi-agent approach,'' \emph{IEEE Transactions on
  Systems, Man and Cybernetics: Systems}, vol.~51, no.~7, pp. 4067--4077, Jul.
  2021.

\bibitem{YH21}
Q.~Yu and Z.~Hou, ``Adaptive fuzzy iterative learning control for high-speed
  trains with both randomly varying operation lengths and system constraints,''
  \emph{IEEE Transactions on Fuzzy Systems}, vol.~29, no.~8, pp. 2408--2418,
  Aug. 2021.

\bibitem{LH21}
G.~Liu and Z.~Hou, ``{RBFNN}-based adaptive iterative learning fault-tolerant
  control for subway trains with actuator faults and speed constraint,''
  \emph{IEEE Transactions on Systems, Man and Cybernetics: Systems}, vol.~51,
  no.~9, pp. 5785--5799, Sep. 2021.

\bibitem{HM19}
W.~He, T.~Meng, S.~Zhang, J.-K. Liu, G.~Li, and C.~Sun, ``Dual-loop adaptive
  iterative learning control for a {Timoshenko} beam with output constraint and
  input backlash,'' \emph{IEEE Transactions on Systems, Man and Cybernetics:
  Systems}, vol.~49, no.~5, pp. 1027--1038, May 2019.

\bibitem{FL21}
J.~Feng, Z.~Liu, X.~He, Q.~Li, and W.~He, ``Vibration suppression of a
  high-rise building with adaptive iterative learning control,'' \emph{IEEE
  Transactions on Neural Networks and Learning Systems}, vol.~34, no.~8, pp.
  4261--4272, Aug. 2023.

\bibitem{CHX2008}
R.~Chi, Z.~Hou, and J.~Xu, ``Adaptive {ILC} for a class of discrete-time
  systems with iteration-varying trajectory and random initial condition,''
  \emph{Automatica}, vol.~44, no.~8, pp. 2207--2213, Aug. 2008.

\bibitem{Goodwin}
G.~C. Goodwin and K.~S. Sin, \emph{Adaptive Filtering Prediction and
  Control}.\hskip 1em plus 0.5em minus 0.4em\relax Englewood Cliffs, NJ:
  Prentice Hall, 1984.

\bibitem{CSH2008}
R.~Chi, S.-L. Sui, and Z.-H. Hou, ``A new discrete-time adaptive {ILC} for
  nonlinear systems with time-varying parametric uncertainties,'' \emph{Acta
  Automatica Sinica}, vol.~34, no.~7, pp. 805--808, Jul. 2008.

\bibitem{SLH2012}
M.~Sun, X.~Liu, and H.~He, ``Adaptive iterative learning control for {SISO}
  discrete time-varying systems,'' in \emph{2012 12th International Conference
  on Control, Automation, Robotics \& Vision}, Guangzhou, China, Dec. 2012, pp.
  58--63.

\bibitem{LZ2015}
B.~Liu and W.~Zhou, ``On adaptive iterative learning control algorithm for
  discrete-time systems with parametric uncertainties subject to second-order
  internal model,'' in \emph{10th International Conference on Computer Science
  \& Education (ICCSE 2015)}, Cambridge, United Kingdom, Jul. 2015, pp.
  512--517.

\bibitem{YZQ2012}
M.~Yu, J.~Zhang, and D.~Qi, ``Discrete-time adaptive iterative learning control
  with unknown control directions,'' \emph{International Journal of Control,
  Automation and Systems}, vol.~10, no.~6, pp. 1111--1118, Dec. 2012.

\bibitem{YWQ2013}
M.~Yu, J.~Wang, and D.~Qi, ``Discrete-time adaptive iterative learning control
  for high-order nonlinear systems with unknown control directions,''
  \emph{International Journal of Control}, vol.~86, no.~2, pp. 299--308, Feb.
  2013.

\bibitem{YS2012}
W.~Yan and M.~Sun, ``Adaptive iterative learning control of discrete-time
  varying systems with unknown control direction,'' \emph{International Journal
  of Adaptive Control and Signal Processing}, vol.~27, no.~4, pp. 340--348,
  Apr. 2012.

\bibitem{YHH2016}
M.~Yu, D.~Huang, and W.~He, ``Robust adaptive iterative learning control for
  discrete-time nonlinear systems with both parametric and nonparametric
  uncertainties,'' \emph{International Journal of Adaptive Control and Signal
  Processing}, vol.~30, no.~7, pp. 972--985, Jul. 2016.

\bibitem{YL2017}
M.~Yu and C.~Li, ``Robust adaptive iterative learning control for discrete-time
  nonlinear systems with time-iteration-varying parameters,'' \emph{IEEE
  Transactions on Systems, Man and Cybernetics: Systems}, vol.~47, no.~7, pp.
  1737--1745, Jul. 2017.

\bibitem{BH18}
X.~Bu and Z.~Hou, ``Adaptive iterative learning control for linear systems with
  binary-valued observations,'' \emph{IEEE Transactions on Neural Networks and
  Learning Systems}, vol.~29, no.~1, pp. 232--237, Jan. 2018.

\bibitem{YHPD21}
X.~Yu, Z.~Hou, M.~M. Polycarpou, and L.~Duan, ``Data-driven iterative learning
  control for nonlinear discrete-time {MIMO} systems,'' \emph{IEEE Transactions
  on Neural Networks and Learning Systems}, vol.~32, no.~3, pp. 1136--1148,
  Mar. 2021.

\bibitem{YH20}
Q.~Yu, Z.~Hou, X.~Bu, and Q.~Yu, ``{RBFNN}-based data-driven predictive
  iterative learning control for nonaffine nonlinear systems,'' \emph{IEEE
  Transactions on Neural Networks and Learning Systems}, vol.~31, no.~4, pp.
  1170--1182, Apr. 2020.

\bibitem{NB1992}
K.~S. Narendra and J.~Balakrishnan, ``Performance improvement in adaptive
  control systems using multiple models and switching,'' in \emph{Proceedings
  of the Seventh Yale Workshop on Adaptive and Learning Systems}, Center for
  Systems Science, Yale University, New Haven, CT, USA, May 1992, pp. 27--33.

\bibitem{MM-CT}
------, ``Adaptive control using multiple models,'' \emph{IEEE Transactions on
  Automatic Control}, vol.~42, no.~2, pp. 171--187, Feb. 1997.

\bibitem{MM-DT}
K.~S. Narendra and C.~Xiang, ``Adaptive control of discrete-time systems using
  multiple models,'' \emph{IEEE Transactions on Automatic Control}, vol.~45,
  no.~9, pp. 1669--1686, Sep. 2000.

\bibitem{MG2021}
R.~Makam and K.~George, ``Convex combination of multiple models for
  discrete-time adaptive control,'' \emph{International Journal of Systems
  Science}, vol.~53, no.~4, pp. 743--756, 2022.

\bibitem{MG2024}
------, ``Multiple models for decentralised adaptive control of discrete-time
  systems,'' \emph{International Journal of Adaptive Control and Signal
  Processing}, pp. 1--16, 2024, early access.

\bibitem{LW2012}
L.~Xiaoli and Z.~Wen, ``Multiple model iterative learning control,''
  \emph{Neurocomputing}, vol.~73, no. 13--15, pp. 2439--2445, Aug. 2012.

\bibitem{LWL2014}
X.~Li, K.~Wang, and D.~Liu, ``An improved result of multiple model iterative
  learning control,'' \emph{IEEE/CAA Journal of Automatica Sinica}, vol.~1,
  no.~3, pp. 315--322, Jul. 2014.

\bibitem{FF2015}
C.~Freeman and M.~French, ``Estimation based multiple model iterative learning
  control,'' in \emph{2015 54th IEEE Conference on Decision and Control (CDC)},
  Osaka, Japan, Dec. 2015, pp. 6070--6075.

\bibitem{PBHMG21a}
R.~Padmanabhan, M.~Bhushan, K.~K. Hebbar, R.~Makam, and K.~George, ``A novel
  strategy with multiple models to improve performance of adaptive iterative
  learning control,'' in \emph{2021 IEEE International Conference on
  Electronics, Computing and Communication Technologies (CONECCT)}, Bengaluru,
  India, Jul. 2021.

\bibitem{PBHMG21b}
------, ``Second-level adaptation and optimization for multiple model adaptive
  iterative learning control,'' in \emph{2021 Seventh Indian Control Conference
  (ICC)}, Mumbai, India, Dec. 2021.

\bibitem{NH2011}
K.~S. Narendra and Z.~Han, ``Discrete-time adaptive control using multiple
  models,'' in \emph{Proceedings of the 2011 American Control Conference}, San
  Francisco, CA, USA, Jun.-Jul. 2011, pp. 2921--2926.

\bibitem{MRG2018}
R.~Makam, S.~Ramaiah, and K.~George, ``Stability analysis of deterministic
  discrete-time adaptive systems with second level adaptation,'' in \emph{2018
  International Conference on Signals and Systems (ICSigSys)}, Bali, Indonesia,
  May 2018, pp. 167--173.

\bibitem{LN1986}
T.-H. Lee and K.~S. Narendra, ``Stable discrete adaptive control with unknown
  high-frequency gain,'' \emph{IEEE Transactions on Automatic Control},
  vol.~31, no.~5, pp. 477--479, May 1986.

\bibitem{QGX2022}
Y.~Qi, H.~Geng, and N.~Xing, ``Adaptive learning control for triggered switched
  systems based on unknown direction control gain function,''
  \emph{International Journal of Control}, vol.~97, no.~2, pp. 175--186, Feb.
  2024.
}

\end{thebibliography}
\end{document}